\documentclass[a4paper,UKenglish,cleveref, autoref, thm-restate]{lipics-v2021}

\pdfoutput=1 %
\hideLIPIcs  %

\usepackage{svg}

\usepackage{booktabs} %

\usepackage{bm}
\usepackage[
algo2e,
ruled, vlined,
linesnumbered,
norelsize]{algorithm2e}
\DontPrintSemicolon
\SetAlCapHSkip{0pt}
\IncMargin{-\parindent}
\SetProgSty{}
\SetFuncSty{}
\SetArgSty{}

\SetKwComment{tcp}{{\upshape{}/\!/}\,}{}%
\SetCommentSty{textit}
\SetKwProg{Function}{Function}{}{}

\usepackage{color, colortbl} %
\usepackage{xcolor}
\definecolor{my-dark-red}{RGB}{183, 28, 28}
\definecolor{my-red}{RGB}{244,67,54}
\definecolor{my-pink}{RGB}{233,30,99}
\definecolor{my-purple}{RGB}{156,39,176}
\definecolor{my-deep-purple}{RGB}{103,58,183}
\definecolor{my-indigo}{RGB}{63,81,181}
\definecolor{my-blue}{RGB}{33,150,243}
\definecolor{my-light-blue}{RGB}{3,169,244}
\definecolor{my-cyan}{RGB}{0,188,212}
\definecolor{my-teal}{RGB}{0,150,136}
\definecolor{my-green}{RGB}{76,175,80}
\definecolor{my-light-green}{RGB}{139,195,74}
\definecolor{my-lime}{RGB}{205,220,57}
\definecolor{my-yellow}{RGB}{255,235,59}
\definecolor{my-amber}{RGB}{255,193,7}
\definecolor{my-orange}{RGB}{255,152,0}
\definecolor{my-deep-orange}{RGB}{255,87,34}
\definecolor{my-brown}{RGB}{121,85,72}
\definecolor{my-grey}{RGB}{158,158,158}
\definecolor{my-blue-grey}{RGB}{96,125,139}
\definecolor{my-lipics-grey}{rgb}{0.6,0.6,0.61}

\usepackage{tikz}
\usetikzlibrary {
  backgrounds,%
  calc,%
  positioning,%
  shapes.geometric,%
  shapes.multipart,%
  decorations.pathreplacing,
}

\usepackage{enumitem}

\usepackage{algorithm}
\usepackage{algpseudocode}
\usepackage{tcolorbox}

\usepackage{mathtools}
\DeclarePairedDelimiter\ceil{\lceil}{\rceil}
\DeclarePairedDelimiter\floor{\lfloor}{\rfloor}

\usepackage{tabularx}
\newcolumntype{Y}{>{\centering\arraybackslash}X}

\title{Theory Meets Practice for Bit Vectors Supporting Rank and Select}

\author{Florian Kurpicz}{Karlsruhe Institute of Technology, Germany}{kurpicz@kit.edu}{https://orcid.org/0000-0002-2379-9455}{}
\author{Niccolò Rigi-Luperti}{Karlsruhe Institute of Technology, Germany}{niccolo.rigi-luperti@kit.edu}{https://orcid.org/0009-0009-6649-9071}{}
\author{Peter Sanders}{Karlsruhe Institute of Technology, Germany}{sanders@kit.edu}{https://orcid.org/0000-0003-3330-9349}{}

\authorrunning{F. Kurpicz, N. Rigi-Luperti, and P. Sanders}

\supplementdetails[subcategory={Source code}]{Software}{https://github.com/nrilu/3star-rankselect/}


\funding{\flag{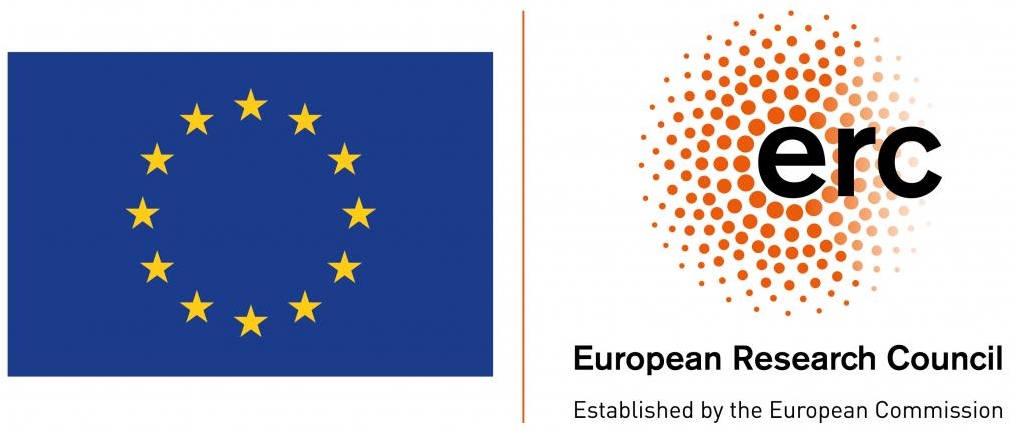}
This project has received funding from the European Research Council (ERC) under the European Union’s Horizon 2020 research and innovation programme (grant agreement No. 882500).}


\nolinenumbers %

\usepackage{amsopn}

\usepackage{tcolorbox}

\newcommand{\frage}[1]{#1}
\renewcommand{\frage}[1]{}

\newcommand{\nr}[1]{\frage{\textcolor{orange}{[NR: #1]}}}
\newcommand{\ps}[1]{\frage{\textcolor{blue}{[PS: #1]}}}

\newcommand{\Id}[1]{{\emph{#1}}}

\newcommand{\Nmax}{\ensuremath{\alpha}}
\newcommand{\astar}{$a^\star$}
\newcommand{\bstar}{$b^\star$}

\newcommand{\rank}[2]{\ensuremath{\mathit{rank}_{#1}(#2)}}
\newcommand{\select}[2]{\ensuremath{\mathit{select}_{#1}(#2)}}

\begin{document}

\maketitle

\begin{abstract}
  Bit vectors with support for fast rank and select
  are a fundamental building block for compressed data structures. We close a gap between theory and practice by analyzing an important part of the design space and experimentally evaluating a sweet spot. The result is the first implementation of a rank and select data structure for bit vectors with worst-case constant query time, good practical performance, and a space-overhead of just 0.78\,\%, i.e., between \(4.5\times\) and \(64.1\times\) less than previous implementations.
\end{abstract}

\section{Introduction}
\label{sec:introduction}

Given a bit vector $v[0..n)\in\{\mathtt{0},\mathtt{1}\}^n$, we want to support queries
\begin{align*}\textstyle
     \rank{\mathtt{1}}{i}&=\overbrace{\textstyle\sum_{j=0}^{i-1}v[j]}^{\makebox[0cm][c]{\text{\scriptsize \texttt{1}-bits up to position $i$}}},
     \rank{\mathtt{0}}{i}=i-\rank{1}{i}\text{, and}\\
     \select{x}{i}&=\underbrace{\arg\textstyle\min_j(\rank{x}{j}=i)-1}_{\makebox[0cm][c]{\text{\scriptsize $i$-th $x$-bit}}}\text{ for }x\in\{\mathtt{0},\mathtt{1}\}.
\end{align*}
  
Such static  bit vectors are fundamental building blocks for succinct, compressed, and compact data structures.
For example, a tree with \(n\) nodes can be encoded as a $2n$-bit vector \cite{Jacobson1989LOUDS} such that,
using rank/select, many operations like navigation in the tree can be supported.
A sorted sequence of \(n\) integers from a universe of size \(u\) can be represented using \(n\cdot\ceil{2+\log\frac{n}{u}}\) bits of space~\cite{Elias1974EliasFano,Fano1971EliasFano} and then efficiently accessed via rank- and select queries.
Bit vectors also have applications in full-text indexing~\cite{FerraginaM2000FMIndex,GagieNP2020FullyFunctionalRIndex}, computational geometry~\cite{DinklageEFKL2021PracticalWaveletTrees,FerraginaGM2009MyriadWT,Makris2012WaveletSurvey,Navarro2014WaveletForAll} and data analysis~\cite{BrodalGJS11} in the form of wavelet trees~\cite{GrossiGV2003WaveletTree}.

From a theoretical perspective, bit vectors are well understood. Rank and select queries can be supported in constant time~\cite{ClarkM1996Select,Jacobson1989LOUDS} provided we are willing to spend additional space $\Theta(n\log\log n/\log n)$ \cite{Golynski2007LowerBound}.
There are practical implementations of \Id{rank} that closely follow the theoretical considerations and nicely align with architectural features like vector instructions \cite{GonzalezGMN2005PracticalRankSelect,GogBMP2014SDSL,Vigna2008BroadwordRankSelect}. However, there is a large gap between theory an practice when we want to support \Id{select}.
The best theoretical constructions~\cite{Golynski2007LowerBound,RamanRS2007RRR} do not look promising for a direct implementation.
They are vague or only achieve low space overhead for superastronomical input sizes. Practical implementations \cite{Vigna2008BroadwordRankSelect} have considerable space overhead (more than needed for \Id{rank}) both in concrete numbers and when extrapolating into asymptotic behavior.
They also exhibit non-constant worst-case query time~\cite{Kurpicz2022PastaFlat,ZhouAK2013PopcountRankSelect}.

In this paper, we close this gap.
We first generalize known approaches for rank and select data structures into a neat design space with different combinations of \emph{summary trees} (counting the number of \texttt{1}-bits in increasingly large chunks) and \emph{sample trees} (storing positions of some \texttt{1}-bits). We then perform a generic analysis
of space consumption as a function of several tuning parameters. Optimizing these parameters then yields both asymptotic bounds improving the state-of-the-art and good starting points for a practical implementation.
We then efficiently implement identified sweet spots using vector instructions and bit-compression techniques. The resulting data structure has much lower space-overhead than previous implementations.
Its space-overhead is just 0.78\,\% and it requires between \(4.5\times\) and \(64.1\times\) less memory while achieving similar or better query performance compared to previous implementations that support both rank and select and sometimes just select.

\noindent{\bf Summary of Contributions:}
\begin{itemize}[noitemsep,topsep=0pt,parsep=0pt,partopsep=0pt]
\item In \cref{sec:design_and_analysis}, we unfold and analyze a design space of select data structures on top of a rank data structure
  that achieve space overhead $o(\text{size of rank data structure})$,
  and construction time $o(\text{time to construct the rank data structure})$.
\item In \cref{s:ThreeStar}, we exemplify a methodology for algorithm
  analysis and parameter tuning that yields both
  improved asymptotic bounds and useful concrete
  values for a practical implementation.
\item We show the effective use of vector instructions and bit-compression techniques in \cref{sec:implementation}.
\item We present an efficient implementation of a sweet spot of this design space that has 0.78\,\% space overhead
  and comparable query time to previous
  implementation with much larger space overhead, see \cref{sec:experimental_evaluation}.
\end{itemize}

\section{Preliminaries}
\label{sec:preliminaries}
Throughout this paper, we assume \(v\) to be a bit vector of length \(n\) containing \(m\) \texttt{1}-bits.
We use $\log x$ for the base-2 logarithm $\log_2 x$ and
$i..j$ is a shorthand for the integer range $\{i,\ldots,j\}$.

We analyze algorithms in the random access machine (RAM) model \cite{Sheperson-Sturgis}.
In particular, a wide range of operations can be executed in constant time on
$O(\log n)$ bits. This \emph{bit/word-parallelism} is ``boosted'' in practice
by the availability of vector registers with up to 512 bits.
This could be modeled by providing an additional word-length model parameter $W$
but we wanted to keep the model simple.
We frequently use \emph{bit-compression}, i.e., the most elementary way to compress integers in the range $0..U-1$, by representing them as $\ceil{\log U}$-bit binary numbers.

\section{Rank-Select Data Structures}
\label{sec:design_and_analysis}
In this section, we first generalize previous rank-select data structures using the notions of summary and sample trees (\cref{ss:designSpace}).
Sample trees are used for selection, while summary trees can be used for both rank and select.
We analyze space consumption and construction time of various implementations in \cref{ss:analysis}.

\subsection{Sketching the Design Space}\label{ss:designSpace}

We now outline a space of designs for rank and select data structures that can be instantiated in many ways. We begin with presenting the rank structure, and then discuss extensions to also support select queries.

\paragraph*{Level-0 Blocks}
The input bit-vector can be subdivided into Level-0 (L0) Blocks of size $L_0=\Theta(\log n)$.
Using word parallelism, rank and select operations can be performed on those blocks in constant time.

\paragraph*{Level-1 Blocks}
To achieve a fast and space-efficient rank data structure, one can
group $d_1$ L0-blocks at a time into one L1-block of size $L_1=d_1L_0$.
The $i$-th L1-block stores the total number $s_i$ of \texttt{1}-bits in L1-blocks $0..i-1$ and
information about the number of \texttt{1}-bits in its L0-blocks.
The concrete representation of L1-blocks is implementation dependent.
We describe our implementation in \cref{summary-tree-implementation}. 
Compressed rank data structures, use the L1-blocks to perform rank operations in constant time while minimizing space consumption.

\begin{lemma}\label{lem:rank}
  For $L_0=O(\log n)$ and $d_1=O(\log n/\log\log n)$, rank can be implemented in constant time. Moreover, the data structure can be generated in time $O(n/\log n)$ and takes space $O(n\log\log n/\log n)$.
\end{lemma}
\begin{proof}
  We explain only rank$_1$, as rank$_0(i)=i-\mathrm{rank}_1(i)$.
  First note that for $L_0=O(\log n)$ and $d_1=O(\log n/\log\log n)$ the information needed for the L1-block
  can be stored in $\log n + d_1\log L_0=O(\log n)$ bits.
  
  \emph{Operation rank$_1(i)$:}
  Let $j=\floor{i/L_1}$ denote the L1-block containing $v[i]$.
  Let \ps{was $j'=\floor{i/L_0-d_1j}$}$j'=\floor{i/L_0}-d_1j$ denote the L0-block containing $v[i]$ within L1-block $j$.
  The result is the sum of three values: The precomputed number of \texttt{1}-bits in L1-blocks $0..j$; the precomputed number of \texttt{1}-bits in the L0-blocks $0..j'-1$ and the
  number of \texttt{1}-bits in bits $0..i\bmod L_0$ of L0-block $j'$.
  The latter can be computed by masking out bits $i\bmod L_0+1..L_0-1$ followed
  by a population count instruction.
  
  \emph{Construction:} An L$1$-block is constructed by scanning the $d_1$ L$0$-blocks it summarizes.
  For each of those, the number of \texttt{1}-bits can be found in constant time
  using population-count instructions. 
\end{proof}

\paragraph*{Select on L1-Blocks}
Say we are given the query select$_1(i)$. Suppose we know that the $i$-th \texttt{1}-bit is in L1-block $B$, and that before $B$ there are $u$ \texttt{1}-bits. Then the query can be reduced to select$_1(i-u)$ within $B$.
It turns out that the information stored within the rank data structure can also be used to perform this task.

\paragraph*{Summary trees}
\ps{more detail on how we actually compute select (1 and 0)?}We can generalize this observation
by using $k$ levels with block sizes $L_j=d_1 \cdot\ldots\cdot d_{j}L_0$.
Here, we have \emph{superblocks} of size $L:= L_{k-1}$.
A select query within a superblock descends the corresponding summary tree
to perform a select operation.
This takes constant time and requires negligible additional space to the data already needed for rank operations.
\begin{lemma}\label{lem:summarySelect}
	Using a $k$-level summary tree with $k\in O(1)$, select$_{1/0}$ within a superblock of size
  $L$ can be performed in constant time if 
  $\{d_1,\ldots,d_j\}\subseteq O(\log n/\log\log n)$. Moreover, the data structure can be generated in time $O(n/\log n)$. The space consumptions of levels L2, \ldots, L$j$ is $O(n(\log\log n/\log n)^2)$.
\end{lemma}
\begin{proof}
  We store the prefix sum of the number of \texttt{1}-bits in the blocks on the next lower level.
  There are \(O(\log n/\log\log n)\) such blocks.
  The maximum value, i.e., the number of covered bits is at most \(O((\log n)^{k+1}/(\log\log n)^{k})\) and can be stored in \(O(\log\log n)\) bits.
  The prefix sum requires \(O(\log n)\) bits of space in total, which we can process in constant time (per level).
  There are at most  \(O(n(\log\log n)^2/(\log n)^3)\) blocks (for L2), requiring in total \(O(n(\log\log n/\log n)^2)\) bits of space.
  Since we compute all this information based on the rank data structure, it requires \(O(n/\log n)\) time.
\end{proof}

\begin{figure}[t]
  \centering
  \begin{tikzpicture}[my-array/.style={
      rectangle split, rectangle split parts=#1, draw, rectangle split horizontal,minimum height=.5cm,text width=.25cm,text badly centered},%
    fill_node/.style 2 args={draw=#1,fill=#1,rectangle,minimum height=#2, minimum width=#2,inner sep=0cm},
    dot/.style={draw,fill=white, circle, inner sep = 0pt, minimum size = 4pt}]

    \node[my-array={3}] (top_level) {\nodepart{one}\nodepart{two}\nodepart{three}\nodepart{four}};

    \node[my-array={4}, below right=.5cm and -1cm of top_level] (mid_level) {\nodepart{one}\nodepart{two}\nodepart{three}\nodepart{four}};

    \node[my-array={4}, below right=.5cm and -1cm of mid_level] (bot_level) {\nodepart{one}\nodepart{two}\nodepart{three}\nodepart{four}};

    \node[right=.01cm of top_level] {\small \(\dots\)};
    \node[right=.01cm of mid_level] {\small \(\dots\)};
    \node[right=.01cm of bot_level] {\small \(\dots\)};

    \node[my-array={6},text width=.975cm,text badly centered] (summary_tree) at (0,-3.25) {%
      \nodepart{one}\(\underset{1|3|6}{0}\) %
      \nodepart{two} \(\underset{8|8|8}{8}\)%
      \nodepart{three} \(\underset{12|12|13}{12}\)%
      \nodepart{four} \(\underset{17|18|18}{15}\)%
      \nodepart{five} \(\underset{19|19|19}{19}\)%
      \nodepart{six} \(\underset{22|22|22}{19}\)};

    \node[right=.01cm of summary_tree] {\small \(\dots\)};

    \path[draw,-latex] (top_level.base) to[out=270,in=90] ($(summary_tree.one)+(0,.325)$);
    \node[dot] at (top_level.base) {};

    \path[draw,-latex] ($(top_level.two)+(.125,0)$) to[out=270,in=90] ($(mid_level.one)+(0,.275)$);
    \node[dot] at ($(top_level.two)+(.125,0)$) {};

    \path[draw,-latex] ($(mid_level.one)+(.125,0)$) to[out=270,in=90] ($(summary_tree.two)+(0,.325)$);
    \node[dot] at ($(mid_level.one)+(.125,0)$) {};
    \path[draw,-latex] ($(mid_level.two)+(.125,0)$) to[out=270,in=90] ($(summary_tree.three)+(0,.325)$);
    \node[dot] at ($(mid_level.two)+(.125,0)$) {};

    \path[draw,-latex] ($(top_level.three)+(.125,0)$) to[out=270,in=90] ($(mid_level.three)+(0,.275)$);
    \node[dot] at ($(top_level.three)+(.125,0)$) {};

    \path[draw,-latex] ($(mid_level.three)+(.125,0)$) to[out=270,in=90] ($(bot_level.one)+(0,.275)$);
    \node[dot] at ($(mid_level.three)+(.125,0)$) {};

    \path[draw,-latex] ($(bot_level.one)+(.125,0)$) to[out=270,in=90] ($(summary_tree.four)+(0,.325)$);
    \node[dot] at ($(bot_level.one)+(.125,0)$) {};
    \path[draw,-latex] ($(bot_level.two)+(.125,0)$) to[out=270,in=90] ($(summary_tree.four)+(0,.325)$);
    \node[dot] at ($(bot_level.two)+(.125,0)$) {};
    \path[draw,-latex] ($(bot_level.three)+(.125,0)$) to[out=270,in=90] ($(summary_tree.five)+(0,.325)$);
    \node[dot] at ($(bot_level.three)+(.125,0)$) {};
    \path[draw,-latex] ($(bot_level.four)+(.125,0)$) to[out=270,in=90] ($(summary_tree.six)+(0,.325)$);
    \node[dot] at ($(bot_level.four)+(.125,0)$) {};

    \path[draw,-latex] ($(mid_level.four)+(.125,0)$) to[out=270,in=90]
    ($(mid_level.four)+(1.5,-.95)$) to[out=270,in=90] ($(summary_tree.six)+(0,.325)$);
    \node[dot] at ($(mid_level.four)+(.125,0)$) {};

    \node[my-array={6},text width=.975cm,text badly centered, below=1.5cm of summary_tree.west, anchor=west,align=center,draw=my-lipics-grey!50] (bv) {%
      \nodepart{one}\tiny \colorbox{my-teal!10}{\texttt{00010000}}\\\tiny\colorbox{my-teal!10}{\texttt{01000100}}\\\tiny\colorbox{my-teal!10}{\texttt{00100101}}\\\tiny\colorbox{my-teal!10}{\texttt{00110000}}\\ %
      \nodepart{two} \tiny\colorbox{my-teal!10}{\texttt{00000000}}\\\tiny\colorbox{my-teal!10}{\texttt{00000000}}\\\tiny\colorbox{my-teal!10}{\texttt{00000000}}\\\tiny\colorbox{my-teal!10}{\texttt{01110010}}\\%
      \nodepart{three} \tiny\colorbox{my-teal!10}{\texttt{00000000}}\\\tiny\colorbox{my-teal!10}{\texttt{00000000}}\\\tiny\colorbox{my-teal!10}{\texttt{00000100}}\\\tiny\colorbox{my-teal!10}{\texttt{00100100}}\\%
      \nodepart{four} \tiny\colorbox{my-teal!10}{\texttt{01000001}}\\\tiny\colorbox{my-teal!10}{\texttt{10000000}}\\\tiny\colorbox{my-teal!10}{\texttt{00000000}}\\\tiny\colorbox{my-teal!10}{\texttt{10000000}}\\%
      \nodepart{five} \tiny\colorbox{my-teal!10}{\texttt{00000000}}\\\tiny\colorbox{my-teal!10}{\texttt{00000000}}\\\tiny\colorbox{my-teal!10}{\texttt{10000000}}\\\tiny\colorbox{my-teal!10}{\texttt{00000010}}\\%
      \nodepart{six} \tiny\colorbox{my-teal!10}{\texttt{10100100}}\\\tiny\colorbox{my-teal!10}{\texttt{00000000}}\\\tiny\colorbox{my-teal!10}{\texttt{00000000}}\\\tiny\colorbox{my-teal!10}{\texttt{00000000}}\\};

    \node[right=.01cm of bv] {\small \(\dots\)};

    \begin{scope}[on background layer]
      \fill[my-lipics-grey!10] (-3.8,-2.45) rectangle ++(8.3,2.9);
      \fill[my-lipics-grey!10] (-3.8,-5.7) rectangle ++(8.3,3.1);

      \node[fill=my-grey,anchor=north west] at (3,.45) {\textcolor{white}{Sample}};
      \node[fill=my-grey,anchor=west] at (3,-2.5) {\textcolor{white}{Summary}};

      \fill[my-grey!50] ($(bv.north west)!1!(bv.north east)$) -- ($(summary_tree.six)+(.45,-.45)$) -- ($(bv.north west)!.84!(bv.north east)$) -- cycle;
      \fill[my-grey!50] ($(bv.north west)!.83!(bv.north east)$) -- ($(summary_tree.five)+(.45,-.45)$) -- ($(bv.north west)!.67!(bv.north east)$) -- cycle;
      \fill[my-grey!50] ($(bv.north west)!.66!(bv.north east)$) -- ($(summary_tree.four)+(.45,-.45)$) -- ($(bv.north west)!.50!(bv.north east)$) -- cycle;

      \fill[my-grey!50] ($(bv.north west)!.49!(bv.north east)$) -- ($(summary_tree.three)+(.45,-.45)$) -- ($(bv.north west)!.33!(bv.north east)$) -- cycle;
      \fill[my-grey!50] ($(bv.north west)!.32!(bv.north east)$) -- ($(summary_tree.two)+(.45,-.45)$) -- ($(bv.north west)!.16!(bv.north east)$) -- cycle;
      \fill[my-grey!50] ($(bv.north west)!.15!(bv.north east)$) -- ($(summary_tree.one)+(.45,-.45)$) -- ($(bv.north west)!0!(bv.north east)$) -- cycle;

    \end{scope}
  \end{tikzpicture}
  \caption{Summary tree with \(L_0=8\) and \(L_1=32\). Large values correspond to \texttt{1}-bits before the L0-block and small values are number of \texttt{1}-bits up to and including the L1-block within the L0-block.
    The sample tree has parameters \(a_0=8\), \(a_1=4\), \(\alpha=2\).
    At the bottom, we give the bit vector.
    Each L1-block is contained in an gray box.
    Each line corresponds to an L0-block.}
  \label{fig:example_sample_summary_tree}
\end{figure}
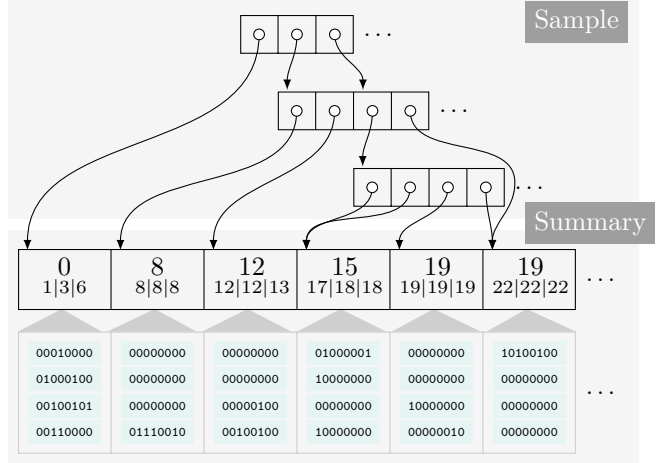

\paragraph*{Sample Trees}
Summary trees alone do not suffice for achieving constant time select because $\Omega(\log n/\log\log n)$ summary-tree levels would be needed.
One can therefore augment the bottom-up summary trees (which are based on coarsening the index space of the bit vector) with a top-down \emph{sample tree} that samples the \texttt{1}-bits.
Separate sample trees are needed for select$_0$ operations.
At the top level, a sample tree stores the superblock containing every
$a_0$-th \texttt{1}-bit. From the positions of two subsequent samples, we can infer a range $r$ of
superblocks that must contain the \texttt{1}-bits between these two samples.
If this range is small, e.g., at most $\Nmax=O(1)$ blocks, a select operation targeting range $r$ can be
completed in constant time.
For larger gaps, i.e. $|r|\geq \alpha$, we can use further levels of samples to provide a locally higher sampling resolution. With an $\ell$-level sample tree,
let $a_0>a_1>\cdots>a_{\ell-1}=1$ denote the sample densities on each level.
Since $a_ {\ell-1}=1$, a query that has to descend to the lowest level can directly read-off the 
superblock containing the sought \texttt{1}-bit. 
An important observation is that there can be at most $s=n/(\Nmax L)$ large ranges.

Some select data structures\ps{dropped~\cite{Vigna2008BroadwordRankSelect} as many others do the same, including the theory papers.} use sample trees with an additional case distinction: if 
a range $r$ is sufficiently large, one directly switches to the last level with dense samples, accelerating some queries. 
We eliminate these tuning parameters by a more adaptive approach to compression in \cref{s:ThreeStar}.
See \cref{fig:example_sample_summary_tree} for an example of a summary and sample tree.

\subsection{Analysis}\label{ss:analysis}

We first provide a general analysis of combining $k$-levels summary trees with
$\ell$-levels sample trees for the simple case that samples are not compressed,
i.e., they always use one machine word of size $w$.  This is a reasonable choice as it might be faster and simpler than data structures with more aggressive compression.

\begin{theorem}\label{thm:spaceUncompressed}
  For bit vector of length \(n\), 
  a select data structure with $\ell\in O(1)$ levels of sample tree and
  $k\in O(1)$ levels of summary tree can be implemented to need space and construction time
  $O\left(n\frac{(\log\log n)^{(k-1)(1-1/\ell)}}{(\log n)^{k-1-k/\ell}}\right)$
  on top of the space and construction time needed for a two-level rank data structure.
\end{theorem}
\begin{proof}
We first consider the required storage space as a
function of various parameters.  Let $w=O(\log n)$
denote the number of bits we use to encode a reference to a superblock
or a pointer to a subsample.
At the top level of the sample tree, when there are $m$ \texttt{1}-bits,
we store $m/a_0$ such references in space $wm/a_0$.
In each subsequent level of the sample tree,
there are $\leq s=n/\Nmax L$ large ranges.
This takes space $wsa_{i-1}/a_i$ at level $i$.

We now set $a_i:=(a_{\ell-2})^{\ell-i-1}$.  Note that
$a_{\ell-1}=1$ as required for a constant time query.
We get space $wm/(a_{\ell-2})^{\ell-1}$ at the top level and
space $wsa_{\ell-2}=wna_{\ell-2}/(\Nmax L)$ at the other levels.
We now set $a_{\ell-2}$ such that both values are the same.
Solving $$\frac{wm}{a_{\ell-2}^{\ell-1}}=\frac{wna_{\ell-2}}{\Nmax L} \quad \; \text{ yields }\quad \;  a_{\ell-2}=\sqrt[\ell]{\frac{m\Nmax L}{n}}.$$
The resulting total space consumption is
\begin{equation}\label{eq:g}
  \ell w\sqrt[\ell]{\frac{mn^{\ell-1}}{\Nmax^{\ell-1}L^{\ell-1}}}.
\end{equation}

Using a $k$-level summary tree with $L_0=\Theta(\log n)$ and degrees $d_i=\Theta(\log n/\log\log n)$ entails $L=\Omega((\log n)^{k}/(\log\log n)^{k-1})$.
Substituting this and $w=O(\log n)$ into Equation~(\ref{eq:g}) yields the claimed asymptotic space bound.

Regarding \emph{construction time}, we can construct the sample tree
top-down and left-to-right on each level.
For each level, it suffices to scan the
set of superblocks once -- without looking at the bit patterns they summarize.
However, this scanning time is dominated by the estimate for the number of generated samples.
\end{proof}

By compressing the sample tree, we can further decrease its memory requirement.

\begin{theorem}\label{thm:spaceCompressed}
  For a length-$n$ bit vector,
  a select data structure with a $2$-level (or $3$-level) sample tree using a $k$-level summary tree requires
  $O\left(n\frac{(\log\log n)^{(k-1)/2}}{(\log n)^{(k-1)/2}}\right)$
  (or $O\left(n\frac{(\log\log n)^{(2k-2)/3}}{(\log n)^{(2k-1)/3}}\right)$)
  bits of space on top of the space needed for a 2-level rank data structure. It can be constructed in the same asymptotic time.
\end{theorem}
Refer to \cref{s:ThreeStar} for a proof and details of the case $\ell=3$ and to \cref{s:compressedTwo} for $\ell=2$.
\newcommand{\llol}[2]{{\footnotesize #1}/#2}
\newcommand{\llolbf}[2]{{\bf{\footnotesize #1}/#2}}

\begin{table}[t]
  \caption{\label{tab:asymptotics}Space overhead of select data structures with $\ell$ levels of sample tree and $k$ levels of summary tree. Notation \llol{a}{b} is a shorthand for $O(\frac{(\log\log n)^a}{(\log n)^b})$. Bold entries have the lowest space overhead for given $\ell+k$. The uncompressed entries are from \cref{thm:spaceUncompressed}. The compressed entry are from \cref{thm:spaceCompressed}, except the case for $k=2$, $\ell=3$ which is due to the special case for $k=2$ discussed in \cref{ss:spaceCompressed}.}

  \centering\tabcolsep1.35mm
  \begin{tabular}{clccccc}
    \toprule
    && \multicolumn{5}{c}{\(\bm k\)}\\\cmidrule{3-7}
    &$\bm \ell$~~ & 2 & 3 & 4 & 5 & $\geq 7$\\
    \midrule
    \parbox[t]{2mm}{\multirow{4}{*}{\rotatebox[origin=c]{90}{\footnotesize\textbf{uncompressed}\hspace{.75mm}}}}  &2        &      ---        & \llolbf{1}{0.5}   & \llol{1.5}{1}  & \llol{2}{1.5} & --- \\[2mm]
    &3        & \llol{0.66}{0.33}  & \llolbf{1.33}{1}  & \llolbf{2}{1.66} & --- & --- \\[2mm]
    &4        & \llol{0.75}{0.5}   & \llol{1.5}{1.25} & --- & --- & --- \\[2mm]
    &5        & \llol{0.8}{0.6}    & --- & --- & --- & --- \\ \midrule \parbox[t]{2mm}{\multirow{3}{*}{\rotatebox[origin=c]{90}{\footnotesize\textbf{compressed}\hspace{-.5mm}}}} &2        & \llolbf{0.5}{0.5}     & \llol{1}{1}  & \llol{1.5}{1.5}  & --- & ---\\[2mm]
    &3        & \llolbf{0.51}{1}     & \llolbf{1.33}{1.66}  &  ---&--- & ---\\[2mm]
    &4        & ---                & \llol{0}{1}~\cite{Golynski2007LowerBound} & --- & --- & ---\\
    \bottomrule
  \end{tabular}
\end{table}

\cref{tab:asymptotics} summarizes relevant configurations resulting from 
\cref{thm:spaceUncompressed,thm:spaceCompressed}. We omit non-succinct configurations as well as some ``overkill'' configurations that would have higher query time than other configurations that already
achieve space $o(n\log\log n/\log n)$.

Among the uncompressed configurations, three seem particularly interesting
as they minimize the asymptotic space-overhead for a value $\ell+k$,
i.e., the total number of tree levels to be traversed.
``2+3'' has the smallest number of levels that achieves succinct space overhead.
``3+3'' ``almost'' matches the space overhead of the rank data structure, and
``3+4'' gets the smallest overhead among configurations with 7 overall levels.
Balancing the number of levels of each type seems to be best; extreme cases do not work well.
For $\ell=1$, no succinct configuration exists while for $k<3$, no configuration beats the overhead of the rank data structure.

Among the configurations with compressed samples, we get a similar picture.
The main difference is that optimal configurations use one summary tree level less than comparable uncompressed configurations.

Previous work can be sorted into this picture as it also uses constructs similar to sample trees and summary trees.
Golynski \cite{Golynski2007LowerBound} uses rank information comparable to 2 levels of summary tree and 4 levels of compressed sample trees using complicated asymptotic settings for the tuning parameters. This configuration has both asymptotically higher space overhead and needs one more level than our configuration with $\ell=3$ and $k=2$. Even the uncompressed version with $\ell=4$ and $k=2$ achieves lower space overhead with the same number of levels.
We view this as an indication that our approach of optimizing with flexible tuning parameters has considerable advantages.
Raman et al. \cite{RamanRS2007RRR} use the ``opposite'' corner of the configuration space compared to Golynski \cite{Golynski2007LowerBound}
with a similar effect. Our solution with $\ell=2$ needs considerably fewer
summary tree levels.
Note that in some cases, the construction time of our schemes is not only sublinear but also $o(n/\log n)$, i.e., faster than the time needed for constructing the rank data structure.
Overall, we use the means of asymptotic analysis to get a better picture of what might be good configurations for actual implementations.

\section{More Related Work}
While the original constant time rank data structure~\cite{Jacobson1989LOUDS} requires just \(O(n\log\log n/\log n)\) bits of space in addition to the bit vector, it utilizes two levels in addition to a lookup table.
This results in subpar practical results.
Providing faster and more space-efficient rank (and select) data structures has been an active line of research~\cite{GogBMP2014SDSL,GonzalezGMN2005PracticalRankSelect,KimNKP2005RankAndSelect,NavarroP2012CombinedSampling,Vigna2008BroadwordRankSelect} culminating in data structures just 3.5\,\% space-overhead~\cite{Kurpicz2022PastaFlat,ZhouAK2013PopcountRankSelect}.

The first constant time select data structure~\cite{ClarkM1996Select} requires \(O(n/\log\log n+\sqrt{n}\log n\log\log n)\) bits of space in addition to the bit vector.
It has prohibitively high constants in practice.
The work on practical select data structures shows a clear space-time trade-off ranging from very fast~\cite{Vigna2008BroadwordRankSelect} to space-efficient~\cite{Kurpicz2022PastaFlat,ZhouAK2013PopcountRankSelect}.
For more information on practical rank-select data structures, we refer to our experimental evaluation in \cref{sec:experimental_evaluation}.

While outside the scope of this paper, we also want to mention additional interesting work:
First, there is a lower bound of \(\Omega(\log n/\log\log n)\) amortized time for rank and select queries on dynamic bit vectors~\cite{FredmanS1989WorstCaseDynamicBitVectorRankSelect}.
A length-\(n\) bit vector with \(m\) \texttt{1}-bits requires at least \(S=\ceil{\log\binom{n}{m}}\) bits of space.
The most space-efficient compressed representation of a bit vector with constant time rank and select queries requires \(S+O(n/\textnormal{poly}\log n)\) bits of space~\cite{Patrascu2008Succincter}.
In practice, the RRR encoding is oftentimes used.
Here, a bit vector with constant time rank and select support can be encoded in \(S+O(n\log\log n/\log n)\) bits of space~\cite{RamanRS2007RRR}.

\newcommand{\mTwo}{\textsc{m2}}
\newcommand{\nameDS}{\ensuremath{3^\star}}
\newcommand{\nameDStwo}{\ensuremath{2^\star}}
\newcommand{\ourMax}{\textnormal{max}}
\newcommand{\afast}{a_{\mathrm{fast}}}

\section{\nameDS: 3-Level Bit-Compressed Sample Tree}\label{s:ThreeStar}
We explored a large design space for select data structures in \cref{sec:design_and_analysis}.
We now zoom in on a ``sweet slice'' of this design space.
We start with a high-level description of a 3-level sample tree.
When answering \(\select{1}{i}\), it returns a superblock close to the one containing the \(i\)-th \texttt{1}-bit.
To this end, we sample the superblock containing every \(a\)-th \texttt{1}-bit in the \emph{top-level} \(T\).
Then, the \(i\)-th \texttt{1}-bit is in a superblock in the range \(r= T[j]..T[j+1]\) with \(j=\floor{i/a}\).
If \(|r|<\alpha\) for a tuning parameter \(\alpha\), we scan the superblocks directly.
For \(|r|\geq\alpha\), we store a \emph{mid-level} sample the superblocks of every \(b\)-th \texttt{1}-bit.
We identify a subrange \(r^\prime\) of superblocks and if \(|r^\prime|<\alpha\), we scan them directly.
Otherwise, we sample the superblock of every \texttt{1}-bit not covered by the first two levels in the \emph{bot-level}.

\newcommand{\flOffset}{\ensuremath{o}}
\newcommand{\slOffset}{\ensuremath{{\flOffset^\Delta}}}
\newcommand{\tlOffset}{\ensuremath{{\flOffset^\Delta_\Delta}}}
\newcommand{\queryIdx}{i}

\subsection{Structure of the Sample Tree}
\label{sec:description_summary_tree}
We now describe the structure of the sample tree in a bottom-up fashion, i.e., we start at the last, the bot-level, containing the highest resolution samples on the bit vector.
The reason for this is that some values in higher levels of the tree depend on information on lower levels, resulting in dependencies that are otherwise hard to explain.
The general idea of \nameDS\ remains the same as previously explained, we want to compute a small range of superblocks that contain the sought \texttt{1}-bit.
Our sample tree uses three parameters: the sample rate at the top-level \(a=a_0\), the sample rate at the mid-level \(b=a_1\), and a limit $\alpha$ on the number of superblocks that we are willing to search.
See \cref{fig:example_sample_tree} for an overview of the levels and \cref{tab:notation} for the used symbols.

For compression, we store information packed with as few bits as possible.
We denote the bit-width necessary to store entries at the mid-level by \(\rho\) and at the bottom-level by \(\rho^\prime\).
Similarly, the gap sizes at the top-level and mid-level are denoted by \(r\) and \(r^\prime\), resp.
Also, at the bot-level and mid-level there are families of arrays with different bit-widths \(M_\rho\) and \(B_{\rho^\prime}\).

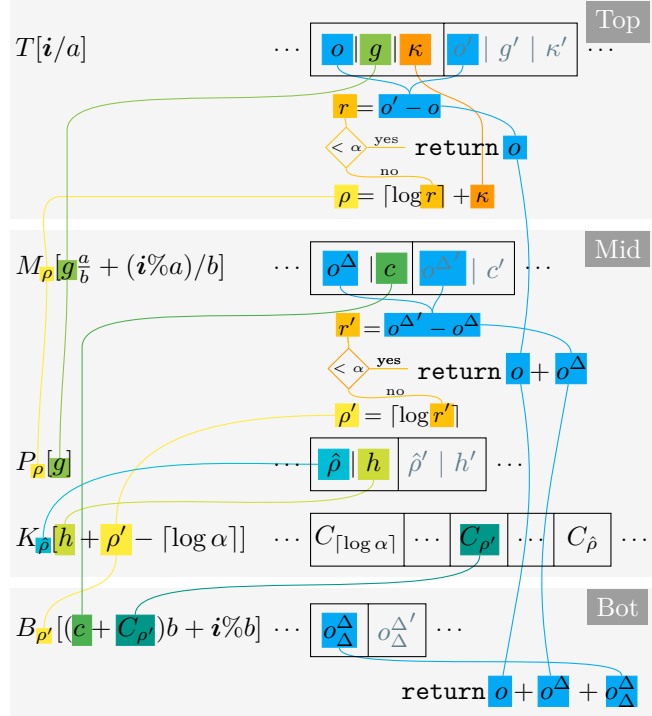
\begin{figure}[t]
  \centering
\newcommand{\textover}[3][l]{%
  \makebox[\widthof{#3}][#1]{#2}%
}
\newcommand{\mathtextover}[3][l]{\mathmakebox[\widthof{\(#3\)}][#1]{#2}}

  \begin{tikzpicture}[my-array/.style={
      rectangle split, rectangle split parts=#1, draw, rectangle split horizontal,minimum height=.7cm},%
    fill_node/.style 2 args={draw=#1,fill=#1,rectangle,minimum height=#2, minimum width=#2,inner sep=0cm}]

    \node (top_label) {\(T[\bm{i}/a]\)};
    
    \node[below=2.9cm of top_label.west, anchor=west] (mid_label) {\(M_\rho[g\frac{a}{b}+(\bm{i}\%a)/b]\)};

    \node[below=2.6cm of mid_label.west, anchor=west] (p_label) {\(P_\rho[g]\)};

    \node[below=1cm of p_label.west, anchor=west] (k_label) {\(K_{\hat{\rho}}[h+\rho^\prime-\ceil{\log\alpha}]\)};

    \node[below=1.2cm of k_label.west, anchor=west] (bot_label) {\(B_{\rho^\prime}[(c+C_{\rho^\prime})b+\bm{i}\%b]\)};

    \node[right=4cm of top_label.west, my-array=2] (top_level) {\(\mathtextover[c]{\flOffset~\vert~g~\vert~\kappa}{\flOffset^\prime~\vert~g^\prime~\vert~\kappa^\prime}\)\nodepart{two}\textcolor{my-blue-grey}{\(\flOffset^\prime~\vert~g^\prime~\vert~\kappa^\prime\)}};

    \node[below right=.75cm and .25cm of top_level.west, anchor=west] (top_calc_one) {\footnotesize\(r=\flOffset^\prime-\flOffset\)};
    \node[below=1.2cm of top_calc_one.west, anchor=west] (top_calc_two) {\footnotesize\(\rho=\ceil{\log r}+\kappa\)};
    \node[below right=.15cm and -1.485cm of top_calc_one,diamond,draw=my-amber,inner sep=0cm] (r_check) {\tiny \(<\alpha\)};
      \node[right=.5cm of r_check] (top_return) {\texttt{return} \flOffset};

    \node[right=4cm of mid_label.west, my-array=2] (mid_level) {\(\mathtextover[c]{\slOffset~\vert~c}{\slOffset^\prime~\vert~c^\prime}\)\nodepart{two}\textcolor{my-blue-grey}{\(\slOffset^\prime~\vert~c^\prime\)}};

    \node[below right=.75cm and .25cm of mid_level.west, anchor=west] (mid_calc_one) {\footnotesize\(r^\prime=\slOffset^\prime-\slOffset\)};
    \node[below=1.2cm of mid_calc_one.west, anchor=west] (mid_calc_two) {\footnotesize\(\rho^\prime=\ceil{\log r^\prime}\)};
        \node[below right=.15cm and -2cm of mid_calc_one,diamond,draw=my-orange,inner sep=0cm] (r_prime_check) {\tiny \(<\alpha\)};
      \node[right=.5cm of r_prime_check] (mid_return) {\texttt{return} \(\flOffset+\slOffset\)};
    
    \node[right=4cm of p_label.west, my-array=2] (p_level) {\(\mathtextover[c]{\hat{\rho}~\vert~h}{\hat{\rho}^\prime~\vert~h^\prime}\)\nodepart{two}\textcolor{my-blue-grey}{\(\hat{\rho}^\prime~\vert~h^\prime\)}};

    \node[right=4cm of k_label.west, my-array=5] (k_level) {\(\mathtextover[c]{C_{\ceil{\log\alpha}}}{C_{\log\alpha\rceil}}\)\nodepart{two}\footnotesize\(\dots\)\nodepart{three}\(C_{\rho^\prime}\)\nodepart{four}\footnotesize\(\dots\)\nodepart{five}\(\mathtextover[c]{C_{\hat{\rho}}}{C_{\rho^\prime}}\)};

    \node[right=4cm of bot_label.west, my-array=2] (bot_level) {\(\mathtextover[c]{\tlOffset}{\tlOffset^\prime}\)\nodepart{two}\(\textcolor{my-blue-grey}{\tlOffset^\prime}\)};

    \node[below right=.125cm and -.45cm of bot_level] (bot_return) {\texttt{return} \(\flOffset+\slOffset+\tlOffset\)};

    \node[left=0cm of top_level] {\footnotesize\(\dots\)};
    \node[right=0cm of top_level] {\footnotesize\(\dots\)};
    \node[left=0cm of mid_level] {\footnotesize\(\dots\)};
    \node[right=0cm of mid_level] {\footnotesize\(\dots\)};
    \node[left=0cm of p_level] {\footnotesize\(\dots\)};
    \node[right=0cm of p_level] {\footnotesize\(\dots\)};
    \node[left=0cm of k_level] {\footnotesize\(\dots\)};
    \node[right=0cm of k_level] {\footnotesize\(\dots\)};
    \node[left=0cm of bot_level] {\footnotesize\(\dots\)};
    \node[right=0cm of bot_level] {\footnotesize\(\dots\)};

    \begin{scope}[on background layer]
      \fill[my-lipics-grey!10] (-.55,-2.25) rectangle ++(8.45,2.9);
      \fill[my-lipics-grey!10] (-.55,-7) rectangle ++(8.45,4.6);
      \fill[my-lipics-grey!10] (-.55,-8.85) rectangle ++(8.45,1.7);

      \node[fill=my-grey,anchor=north] at (7.45,.65) {\textcolor{white}{Top}};
      \node[fill=my-grey,anchor=north] at (7.45,-2.4) {\textcolor{white}{Mid}};
      \node[fill=my-grey,anchor=north] at (7.45,-7.15) {\textcolor{white}{Bot}};

      \node[fill_node={my-yellow}{.3cm}] (rho_top) at ($(top_calc_two)+(-.89,0)$) {};
      \node[fill_node={my-yellow}{.175cm}] (rho_mid) at ($(mid_label)+(-.9725,-.075)$) {};
      \node[fill_node={my-yellow}{.175cm}] (rho_p) at ($(p_label)+(-.09,-.09)$) {};
      \path[draw=my-yellow] (rho_top) to[out=180,in=0]  ($(rho_top)+(-3.75,0)$) to[out=180,in=90] (rho_mid);
      \path[draw=my-yellow] (rho_mid) to[out=270,in=90] (rho_p);

      \node[fill_node={my-orange}{.4cm}] (kappa_top) at ($(top_level)+(-.375,0)$) {};
      \node[fill_node={my-orange}{.3cm}] (kappa_calc) at ($(top_calc_two)+(.9,0)$) {};
      \path[draw=my-orange] (kappa_top) to[out=270,in=90] ($(kappa_calc)+(0,.9)$) to[out=270,in=90] (kappa_calc);

      \node[fill_node={my-amber}{.3cm}] (r_top) at ($(top_calc_one)+(-.575,0)$) {};
      \node[fill_node={my-amber}{.3cm}] (r_calc) at ($(top_calc_two)+(.24,0)$) {};
      \path[draw=my-amber] (r_check) to[out=0,in=180] node[pos=.5,above,yshift=-.1cm] {\tiny yes} (top_return);
      \path[draw=my-amber] (r_top) to[out=270,in=90] (r_check);
      \path[draw=my-amber] (r_check) to[out=270,in=90,looseness=.55] node[pos=.5,above,yshift=-.1cm] {\tiny no} (r_calc);

      \node[fill_node={my-light-blue}{.4cm}] (offset_top) at ($(top_level)+(-1.4,0)$) {};
      \node[fill_node={my-light-blue}{.4cm}] (offset_prime_top) at ($(top_level)+(.25,0)$) {};
      \node[fill_node={my-light-blue}{.3cm},minimum width=.8cm] (offset_minus) at ($(top_calc_one)+(.275,0)$) {};
      \node[fill_node={my-light-blue}{.3cm}] (offset_top_return) at ($(top_return)+(.62,0)$) {};
      \path[draw=my-light-blue] (offset_top) to[out=270, in=90, looseness=.6] (offset_minus);
      \path[draw=my-light-blue] (offset_prime_top) to[out=270, in=90] (offset_minus);
      \path[draw=my-light-blue] (offset_minus) to[out=0,in=90,looseness=.6] (offset_top_return);

      \node[fill_node={my-light-green}{.4cm}] (g_top) at ($(top_level)+(-.9,0)$) {};
      \node[fill_node={my-light-green}{.25cm}] (g_mid) at ($(mid_label)+(-.705,-.025)$) {};
      \node[fill_node={my-light-green}{.25cm}] (g_p) at ($(p_label)+(.2,-.05)$) {};
      \path[draw=my-light-green] (g_top) to[out=270,in=90,looseness=.5] ($(g_mid)+(0,1.85)$) to[out=270,in=90] (g_mid);
      \path[draw=my-light-green] (g_mid) to[out=270,in=90] (g_p);

      \node[fill_node={my-amber}{.3cm}] (r_mid) at ($(mid_calc_one)+(-.8125,0)$) {};
      \node[fill_node={my-amber}{.3cm}] (r_prime_calc) at ($(mid_calc_two)+(.55,0)$) {};
      \path[draw=my-amber] (r_prime_check) to[out=0,in=180] node[pos=.5,above,yshift=-.1cm] {\tiny yes} (mid_return);
      \path[draw=my-amber] (r_mid) to[out=270,in=90] (r_prime_check);
      \path[draw=my-amber] (r_prime_check) to[out=270,in=90,looseness=.55] node[pos=.5,above,yshift=-.1cm] {\tiny no} (r_prime_calc);
      
      \node[fill_node={my-light-blue}{.45cm}] (offset_mid) at ($(mid_level)+(-.95,0)$) {};
      \node[fill_node={my-light-blue}{.45cm},minimum width=.6cm] (offset_prime_mid) at ($(mid_level)+(.4,0)$) {};
      \node[fill_node={my-light-blue}{.3cm},minimum width=1.25cm] (offset_mid_minus) at ($(mid_calc_one)+(.3,0)$) {};
      \node[fill_node={my-light-blue}{.4cm}] (offset_prime_mid_return) at ($(mid_return)+(.9,0)$) {};
      \node[fill_node={my-light-blue}{.4cm},minimum width=.3cm] (offset_mid_return) at ($(mid_return)+(.2,0)$) {};
      \path[draw=my-light-blue] (offset_mid) to[out=270, in=90, looseness=.6] (offset_mid_minus);
      \path[draw=my-light-blue] (offset_prime_mid) to[out=270, in=90] (offset_mid_minus);
      \path[draw=my-light-blue] (offset_mid_minus) to[out=0,in=90,looseness=.6] (offset_prime_mid_return);
      \path[draw=my-light-blue] (offset_top_return) to[out=280,in=80] (offset_mid_return);

      \path[draw=my-amber] (r_prime_check) to[out=0,in=180] node[pos=.5,above,yshift=-.1cm] {\tiny yes} (mid_return);
      \path[draw=my-amber] (r_mid) to[out=270,in=90] (r_prime_check);

      \node[fill_node={my-yellow}{.3cm}] (rho_prime) at ($(mid_calc_two)+(-.7,0)$) {};
      \node[fill_node={my-yellow}{.4cm}] (rho_prime_k) at ($(k_label)+(-.2,0)$) {};
      \node[fill_node={my-yellow}{.19cm}] (rho_prime_bot) at ($(bot_label)+(-1.25,-.075)$) {};
      \path[draw=my-yellow] (rho_prime) to[out=180,in=90] (rho_prime_k);
      \path[draw=my-yellow] (rho_prime_k) to[out=270,in=90,looseness=1.2] (rho_prime_bot);

      \node[fill_node={my-lime}{.4cm}] (h) at ($(p_level)+(-.325,0)$) {};
      \node[fill_node={my-lime}{.4cm},minimum width=.25cm] (h_k) at ($(k_label)+(-.875,0)$) {};
      \path[draw=my-lime] (h) to[out=270,in=90,looseness=.4] (h_k);

      \node[fill_node={my-cyan}{.4cm}] (rho_hat) at ($(p_level)+(-.85,0)$) {};
      \node[fill_node={my-cyan}{.19cm}] (rho_hat_k) at ($(k_label)+(-1.16,-.075)$) {};
      \path[draw=my-cyan] (rho_hat) to[out=180,in=90,looseness=.8] (rho_hat_k);

      \node[fill_node={my-green}{.4cm}] (c) at ($(mid_level)+(-.275,0)$) {};
      \node[fill_node={my-green}{.4cm},minimum width=.25cm] (c_bot) at ($(bot_label)+(-.75,0)$) {};
      \path[draw=my-green] (c) to[out=270,in=90,looseness=.5] ($(c_bot)+(0,3.75)$) to[out=270,in=90] (c_bot);

      \node[fill_node={my-teal}{.4cm}, minimum width=.5cm] (C_rho_prime) at ($(k_level)+(.25,0)$) {};
      \node[fill_node={my-teal}{.4cm}, minimum width=.5cm] (C_rho_prime_bot) at ($(bot_label)+(-.05,0)$) {};
      \path[draw=my-teal] (C_rho_prime) to[out=270,in=90,looseness=.4] (C_rho_prime_bot);

      \node[fill_node={my-light-blue}{.45cm}] (offset_bot) at ($(bot_level)+(-.38,0)$) {};
      \node[fill_node={my-light-blue}{.4cm}] (offset_prime_prime_bot_return) at ($(bot_return)+(1.35,0)$) {};
      \node[fill_node={my-light-blue}{.4cm}] (offset_prime_bot_return) at ($(bot_return)+(.45,0)$) {};
      \node[fill_node={my-light-blue}{.4cm},minimum width=.3cm] (offset_bot_return) at ($(bot_return)+(-.25,0)$) {};
      \path[draw=my-light-blue] (offset_bot) to[out=270,in=90,looseness=.25] (offset_prime_prime_bot_return);
      \path[draw=my-light-blue] (offset_mid_return) to[out=280,in=80] (offset_bot_return);
      \path[draw=my-light-blue] (offset_prime_mid_return) to[out=260,in=100] (offset_prime_bot_return);
    \end{scope}
  \end{tikzpicture}
  \caption{
    Structure of our 3-level sample tree \nameDS\ including data-flow of a \(\select{1}{\bm{i}}\)-query.
    In light gray, we give the next entry of an array.
    Best viewed in color.}
  \label{fig:example_sample_tree}
\end{figure}

\subsubsection*{Bottom-Level}
An \emph{entry} in the bot-level \(B_{\rho^\prime}\) of \nameDS\ contains an offset \tlOffset.
In \(B_{\rho^\prime}\), \(b\) consecutive entries form a \emph{bot-group}, i.e., they are the result of an entry on the mid-level being split because its gap is too big.

\subsubsection*{Middle-Level}
An entry in the mid-level \(M_\rho\) contains an offset \(\slOffset\) and a counter \(c\). A mid-group is then formed by $\frac{a}{b}$ consecutive such entries. To navigate from a mid-entry to its bot-group with bit-width $\rho^\prime$, we have to know the number of preceding bot-groups with the same bit-width \(\rho^\prime\), i.e. the number of preceding mid-level entries with gap sizes in the range \(2^{\rho^\prime}..2^{{\rho^\prime}+1}-1\). To this end, the counter \(c\) stores the number of such preceeding mid-level entries within the mid-group itself.
However, we also have to know the global count of such preceeding entries before the mid-group, denoted $C_{\rho^\prime}$.
To do so, for each mid-group, we store these global counts for all required ${\rho^\prime}$.
Fortunately, only the range \(\ceil{\log\alpha}..\hat{\rho}\) is relevant for a given mid-group, where \(\hat{\rho}\) is the maximum bit-width ${\rho^\prime}$ caused by any split in that mid-group. Thus, only global counts up to  $C_{\hat\rho}$ need to be stored.
We store these in a family of arrays \(K_{\hat{\rho}}\).
Now, we have to map each mid-group to its global counters.
For this, we store for the mid-group the value \(\hat{\rho}\) and the index \(h\) of the start of its mapped counters in \(K_{\hat{\rho}}\). The indirection into $K_{\hat\rho}$ allow us to allocate only the exact amount of space required and still keep control over all entry locations.%
\footnote{These counters can be replaced by a minimal perfect hash function (MPHF). This results in further compression as MPHF require around 1.5 bits per entry \cite{lehmann2025combinedsearchencodingseeds,Mairson1983LowerBoundMPHF,Mehlhorn1982LowerBoundMPHF}. We refrain from this to avoid the overhead of evaluating the MPHF.}

\subsubsection*{Top-Level}
Finally, we are at the top-level.
Here, each entry consists of an offset \(\flOffset\), a group index \(g\), and a bit-width-counter \(\kappa\).
Let \(r\) be the gap of an entry.
If \(r\geq\alpha\), it is split into the \(g\)-th mid-group that requires \(\rho=\ceil{\log r}+\kappa\) bits of space per mid-entry, where \(\kappa=\log\ceil{\max \{c,c^\prime \dots\}}\) is the size needed to store the maximum local counter $c$ in that mid-group.

\subsubsection*{Following a Query}
\Cref{fig:example_sample_tree} illustrates the path of a \(\select{1}{\queryIdx}\)-query in \nameDS\@.
We start at the top-level with \((\flOffset,g,\kappa)=T[\queryIdx/a]\).
Assuming the gap size \(r\) is at least \(\alpha\), we compute the size of the mid-level entries \(\rho=\ceil{\log r}+\kappa\).
On the mid-level, we are now interested in the entry \((\slOffset,c)=M_\rho[\underbrace{g\frac{a}{b}}_{\textnormal{group}}+\underbrace{(\queryIdx
  \%a)/b\vphantom{\frac{a}{b}}}_{\textnormal{\makebox[0mm][c]{entry in group}}}]\).
We again assume that the gap size \(r^\prime\geq\alpha\).
Let \(\rho^\prime=\ceil{\log r^\prime}\).

We now have to identify how many preceding bot-groups with bit-width \(\rho^\prime\) exist. 
To find that global count, we use that \(\rho^\prime\in\ceil{\log\alpha}..\hat{\rho}\), i.e., our $C_{\rho^\prime}$ is the ($\log r^\prime-\log\alpha$)-th entry within the mapped counters: \(C=K_{\hat{\rho}}[\underbrace{h}_{\textnormal{\makebox[0pt]{start}}}+\underbrace{\ceil{\log r^\prime}-\ceil{\log\alpha}}_\textnormal{specific entry}]\).
Finally, with offset at the bot-level \( \tlOffset = B_{\rho^\prime}[\underbrace{(c+C)b}_{\textnormal{group}}+\underbrace{(k\%b)}_{\textnormal{\makebox[0pt]{entry in group}}}]\), we can start searching for the bit in the \((\flOffset+\slOffset+\tlOffset)\)-th superblock.

\subsection{Space Analysis}\label{ss:spaceCompressed}
We first bound the space required by \nameDS\ for arbitrary sample rates \(a\) and \(b\).
We then find values \(a^\star\) and \(b^\star\) that minimize this bound and give an upper bound for the space required by it.
Throughout this section, we consider \(\alpha\in O(1)\).

\begin{lemma}
  \label{lem:with_a_and_b}
  Up to lower order terms, the sample tree \nameDS\ requires \(\nameDS_\textnormal{space}\)
  \begin{align*}
	 & \leq \underbrace{t \left(\log\frac{n}{L} + \log G + \log \log \frac{a}{b}\right)}_{\textnormal{top-level}} + \underbrace{G \left( \log \frac{s}{G} + 1 \right) \log s}_{\textnormal{mid-level}}\\
		       & + \underbrace{G \left( 2\frac{a}{b}\log \frac{\alpha s}{G}+ \log \log \frac{\alpha s}{G} + \log G \right)}_{\textnormal{mid-level}} 
	    + \underbrace{sb \log \alpha\vphantom{\left(\frac{\alpha s}{G}\right)}}_{\textnormal{bot-level}}
\end{align*}
bits of space, where $s =\frac{n}{\alpha L}$, $t=\frac{m}{a}$, and \(G=\min\{s,t\}\).
\end{lemma}

\begin{proof}
  To simplify calculation, we ``pretend'' that storing elements in a range $0..U-1$ can be stored using
  $\log U$ bits; ignoring rounding up to $\lceil\log U\rceil$ as required for bit-compression.
  This is admissible as, in our calculations, a more detailed calculation only adds lower order terms to the
  space consumption.

  \textbf{Top-Level:} There are $t$ top-groups.
  For each, we store the offset of the superblock, of which there are \(n/L\) many, in \(\log\frac{n}{L}\) bits. 
  The mid-group reference $g$ can count at most up to $G$, the number of mid-groups, so we store it in $\log G$ bits. Finally, $\kappa$ stores the number of bits needed for the largest $c$ in the mid-group created by this top-group.
  Since $c < \frac{a}{b}$ (each mid-group has $\frac{a}{b}$ entries), \(\kappa\leq\log\frac{a}{b}\), and we can store it in \(\log\log\frac{a}{b}\) bits.

  \textbf{Mid-Level:}
  The mid-groups require \(\sum_{g=1}^G\frac{a}{b}\rho_g\) bits of space, where \(\rho_g=\log r_g+\log c^\ourMax_g\) is the number of bits necessary to store any entry of the group.
Here, \(r_g\) is the gap of the top-entry split into the mid-group and \(c^\ourMax_g\) is the greatest counter-value within the mid-group \(g\).
Since the logarithm is convex, the first term is maximized if there are many small gaps, i.e., \(r_g\approx\alpha\).
However, we have to be careful.
There are at most \(s=\frac{n}{\alpha L}\) gaps of size \(\alpha\) and at most \(t=\frac{m}{a}\) top-groups. Both constrain the number of mid-groups to \(G=\min\{s,t\}\), since mid-groups can only created by top-groups, and at most $s$ top-groups can ever split. In case of $G=s$, every top-level gap can be minimal, i.e. $r_g = \alpha$. In the case of $G=t$, however, to maximize the sum, we must 
set  $r_g = \alpha \frac{s}{t}$, since space for $s$ bot-level splits passes through $t<s$ mid-groups.
We capture both cases in general by setting $r_g = \alpha \frac{s}{G}$.
Similarly, we maximize the counter-value by allowing for the most splits in each mid-group, i.e., \(c^\ourMax_g=\frac{s}{G}\).
Overall, we need \(G\frac{a}{b}(\log\frac{\alpha s}{G}+\log\frac{s}{G})\leq 2G\frac{a}{b}\log\frac{\alpha s}{G}\) bits of space for the family of arrays \(M_\rho\).

For each mid-group \(g\), we store the maximum \(\rho^\prime\)-value \(\hat{\rho}_g\) and the start \(h\) of the counters \(C_{\ceil{\log \alpha}},\dots,C_{\hat{\rho}_g}\).
This requires \(G\log G+ \sum_{g=1}^G\log\log\max_g\hat{\rho}_g\) bits of space.
The start \(h\) can be at most as big as there are mid-groups and there are \(G\) many.
Interestingly, we have \(\hat{\rho}_g=\frac{\alpha s}{G}\) due to maximizing the space for mid-entries, which does not maximize this term.
Fortunately, it enters only as lower-order term in the space of \(G(\log G +\log\log\frac{\alpha s}{G})\) bits that are necessary for the family of arrays \(P_\rho\).

Finally, for the counters \(C_{\ceil{\log\alpha}},\dots,C_{\hat{\rho}_g}\), we need at most \(\sum_{g=1}^G(\hat{\rho}_g-\ceil{\log\alpha}+1)\log s\) bits of space, as each counter can count at most to $s$.
Since \(\hat{\rho}_g\) depends on the size of the gaps on the mid-level, which can never be greater than the gaps on the top-level, i.e. $r'_g \leq r_g$, the maximum of the sum cannot be greater than what we get by choosing  \(\hat{\rho}_g=\rho_g=\frac{\alpha s}{G}\).
Hence, we need at most \(G(\log\frac{\alpha s}{G}-\log\alpha + 1)\log s\) bits of space for the family of arrays \(K_{\hat{\rho}}\).

  \textbf{Bot-Level:}   When an entry at the mid-level is split (when \(r^\prime \geq \alpha\)), we store \(b\) offsets using \(\log r^\prime\) bits of space, each.
  There can only exist  \(\leq s\) entries in the mid-level that split, thus the bot-level's space is \(\leq b\sum_{i=1}^s\log r^\prime_i\), where \(r^\prime_i\) is the gap of the \(i\)-th split entry in the mid-level.
  This is maximized (convexity of the logarithm) for \(r^\prime_i=\alpha\) for all \(i\) in \(1..s\), resulting in the claimed size.
\end{proof}

Assuming \(G=t\), we can show that \(\nameDS_\textnormal{space}\in O(n\frac{\sqrt{\log\log n}\sqrt{\log\log\log n}}{\log n})\) for the case $k=2$, see~\cref{sec:special_case_space}. However, \(G=t\) is indeed only space-optimal for \(k=2\).
For \(k\geq 3\), we always have \(G=s\) if we want to be space-optimal, which is the only case we consider in the following, due to space constraints and as we get a more general result for arbitrary \(k\).

\begin{lemma}
  To minimize the space in \cref{lem:with_a_and_b} when using a \(k\)-level summary, the optimal choices for \(a\) and \(b\), are:
  \[a^\star=\left(\frac{3m}{\sqrt{2}n}\frac{\alpha}{\log\alpha}L\log\frac{n}{L}\right)^{2/3}\in O\!\left(L\frac{(\log\log n)^{(k-1)/3}}{(\log n)^{(k-2)/3}}\right)\]
  and
  \[b^\star=(2a^\star)^{1/2}\in O\!\left(\sqrt{L}\frac{(\log\log n)^{(k-1)/6}}{(\log n)^{(k-2)/6}}\right).\]
\end{lemma}

\begin{proof}
	To find the \(a\) that minimizes \(\nameDS_\textnormal{space}\), we consider the part of \(\nameDS_\textnormal{space}\) that depends \(a\), i.e., \(f(a)=t\log\frac{n}{L}+t\log s+2sb\log\alpha+t\log s+s\log s=3\log\frac{n}{aL}+\sqrt{8}\frac{\log\alpha}{\alpha}\frac{n}{mL}\sqrt{a}\).
  To minimize it, we solve \(f^\prime(a)=-3\log\frac{n}{a^2L}+\sqrt{8}\frac{\log\alpha}{\alpha}\frac{n}{mL2\sqrt{a}}\overset{!}{=}0\).
  This results in \(a^\star=(\frac{2m}{\sqrt{3}n}\frac{\alpha}{\log\alpha}L\log\frac{n}{L})^{2/3}\in O((L\log\frac{n}{L})^{2/3})\subseteq O(L\frac{\log\frac{n}{L}}{L^{1/3}})\subseteq O\!\left(L\frac{(\log\log n)^{(k-1)/3}}{(\log n)^{(k-2)/3}}\right)\)
  
   To minimize the part of \(\nameDS_\textnormal{space}\) that depends on \(b\), i.e., \(f(b)=(2\frac{a}{b}+b)\log\alpha\), we solve \(f^\prime(b)=\frac{(b^2-2a)\log\alpha}{b^2}\overset{!}{=}0\), which has a solution for \(b^\star=\sqrt{2a}\)
\end{proof}

Now, we can show the minimized space-requirements of \nameDS\ using the optimal values \(a^\star\) and \(b^\star\).
To this end, we use \(\frac{a^\star}{b^\star}\in O(\sqrt{a^\star})\), and \(L=O(\frac{(\log n)^{k}}{(\log\log n)^{k-1}})\). 

\begin{theorem}
  \label{lem:final_space}
  Our 3-level sample tree \nameDS\ using a \(k\)-level summary tree requires \(O\!\left(n\frac{(\log\log n)^{(2k-2)/3}}{(\log n)^{(2k-1)/3}}\right)\) bits of space and can be constructed in the same time on top of the construction time needed for a rank data structure.
\end{theorem}

\newcommand{\proofVSPACE}{\ensuremath{\vphantom{\frac{n}{L}\frac{a^\star}{b^\star}}}}

\begin{proof}
  We bound the \emph{space} of \nameDS\ by considering all addends (\cref{lem:with_a_and_b}) with values \(a^\star\) and \(b^\star\): \(\nameDS_\textnormal{space}\in O(\underbrace{t\log\frac{n}{L}}_{\cref{eq:first_addend}}+\underbrace{t\log s\proofVSPACE}_{\cref{eq:second_addend}}+\underbrace{t\log\log\frac{a^\star}{b^\star}}_{\cref{eq:third_addend}}+\underbrace{s\frac{a^\star}{b^\star}}_{\cref{eq:fourth_addend}}+\underbrace{s\log s\proofVSPACE}_{\cref{eq:fifth_addend}}+\underbrace{sb^\star\proofVSPACE}_{\cref{eq:sixth_addend}})\).
  \begin{align}
    t\log\frac{n}{L}&\in O\!\left(\frac{n}{a^\star}\log\frac{n}{L}\right)\subseteq O\!\left(n\frac{(\log\frac{n}{L})^{1/3}}{L^{2/3}}\right)\label{eq:first_addend}\\
    &\subseteq O\!\left(n\frac{(\log\log n)^{(2k-2)/3}}{(\log n)^{(2k-1)/3}}\right)\nonumber
  \end{align}\\[-1cm]
  \begin{align}
    t\log s &%
              \overset{\cref{eq:first_addend}}{\in} O\!\left(n\frac{(\log\log n)^{(2k-2)/3}}{(\log n)^{(2k-1)/3}}\right)\label{eq:second_addend}
  \end{align}\\[-1cm]
  \begin{align}
    t\log\log\frac{a^\star}{b^\star} &\in O\!\left(t\log\log\sqrt{a^\star}\right) \subseteq O\!\left(t\log\frac{n}{L}\right)\label{eq:third_addend}\\ &\overset{\makebox[0pt]{\scriptsize\cref{eq:first_addend}}}{\subseteq} O\!\left(n\frac{(\log\log n)^{(2k-2)/3}}{(\log n)^{(2k-1)/3}}\right)\nonumber
  \end{align}\\[-1cm]
  \begin{align}
    s\frac{a^\star}{b^\star} &\in O\!\left(s\sqrt{a^\star}\right)\subseteq O\!\left(\frac{n}{L}\sqrt{L\frac{(\log\log n)^{(k-1)/3}}{(\log n)^{(k-2)/3}}}\right)\label{eq:fourth_addend}\\
    &\subseteq O\!\left(n\frac{(\log\log n)^{(2k-2)/3}}{(\log n)^{(2k-1)/3}}\right)\nonumber
  \end{align}\\[-1cm]
  \begin{align}
    s\log s &\in O\!\left(\frac{n}{L}\log\frac{n}{L}\right)\subseteq O\!\left(n\frac{(\log\log n)^{(k-1)}}{(\log n)^{(k-1)}}\right)\label{eq:fifth_addend}\\
    &\subseteq O\!\left(n\frac{(\log\log n)^{(2k-2)/3}}{(\log n)^{(2k-1)/3}}\right)\nonumber
  \end{align}\\[-1cm]
  \begin{align}
    sb^\star &\in O\!\left(s\sqrt{a^\star}\right)\overset{\cref{eq:fourth_addend}}{\subseteq} O\!\left(n\frac{(\log\log n)^{(2k-2)/3}}{(\log n)^{(2k-1)/3}}\right)\label{eq:sixth_addend}
  \end{align}
  
Analogous to the argument in the uncompressed case (\cref{ss:analysis}), the \emph{construction time} of \nameDS\ is bounded by the maximum number of entries in one of its level.
Therefore, we can compute it in time \(O(n(\frac{(\log\log n)^{(2k-2)/3}}{(\log n)^{(2k-1)/3}}))\).
\end{proof}

\section{Implementation of the \nameDS\ Rank-Select Data Structure}
\label{sec:implementation}

\newcommand*\cpp{C\kern-0.2ex\raisebox{0.4ex}{\scalebox{0.8}{+\kern-0.4ex+}}}

We implemented the \nameDS\ sample tree and a configurable $k$-level summary tree for $k \in \{2,3\}$ in \cpp.
We now discuss some of these implementation details.

\subsection{Implementation of the Sample Tree}
Our implementation of the \nameDS\ sample tree is a bit-precise replica of the theoretical presentation.
It fulfills all asymptotic space bounds and the absolute bounds in terms of number of bits.
It consists of the arrays $T$, $M_{\rho}$, $K_{\hat{\rho}}$, and $B_{\rho'}$, for $\rho, \hat{\rho}, \rho' \in 0..64$.
The contents of $P_{\rho}$ are merged directly with the mid-groups in $M_{\rho}$.

We construct the sample tree for fixed parameters $a, b, \alpha$, and $L$.
For $\alpha$, values $\alpha \in 1..64$ are the most promising.
For $a$ and $b$, we consider \astar\ and \bstar\@.
We also tested different values, see \cref{a_b_discussion}.
In any case, $a$ and $b$ are always rounded up to the nearest power of two.
The superblock size $L$ is provided by the summary tree.

For known $a,b,\alpha$, and $L$, the sample tree is constructed in two passes.
First, all values to be stored are extracted from the bit vector.
For each of the different entry-types $o$, $g$, $\kappa$, $\hat{\rho}, h$, and $C$, the largest occurring value is determined, resulting in the corresponding bit-sizes $\beta_o, \beta_g, \beta_{\kappa}, \beta_{\hat{\rho}}, \beta_{h}$, and $\beta_C $ that are globally fixed for all entry-types.
Then, the tree is packed bit-precise in a second pass. 
E.g., each top-group is stored in $\beta_o + \beta_g + \beta_{\kappa}$ bits and each global counter $ C_\rho$ in $\beta_C $ bits.

A query in our implementation follows exactly the flow depicted in \cref{fig:example_sample_tree}.
Bit-precise readout of all packed values is made possible by keeping track of all globally fixed $\beta$-sizes and query-dependent entry sizes $\rho$, $\hat{\rho}$, and $\rho'$.
Division and modulo by $a$ and $b$ are replaced by much cheaper bit-shifts and logical ands.
The rounded-up logarithms are calculated via the fast \texttt{lzcnt} instruction.

\subsection{Implementation of the Summary Tree}
\label{summary-tree-implementation}

\newcommand{\LLzm}{\mathbb{L}_0}
\newcommand{\LLz}{\(\mathbb{L}_0\)}
\newcommand{\LLo}{\(\mathbb{L}_1\)}

We support three L0-block sizes: $L_0 \in \{512, 1024, 2048\}$.
The configuration $L_0=512$ is attractive because the covered bits fit in one cache line.
Larger $L_0$-values result less space-overhead for the summary trees.
\begin{figure}[t]
	\label{L1-block}
\centering
\begin{tikzpicture}[x=0.025cm, y=1cm, every node/.style={font=\small}]
  \def\arrowYOffset{15pt}
  \def\entrysize{\large}
  \def\bitsize{\scriptsize}

  \tikzset{
    cacheentry/.style={
      draw,
      minimum height=1,
      anchor=south west,
      text height=2.25ex,
      text depth=0.5ex
    },
    bventry/.style={
      draw,
      minimum height=1,
      anchor=south west,
      text height=1.0ex,
      text depth=0.25ex
    }
  }
  \newcommand{\mynewbox}[5]{
      \node[cacheentry, minimum width=#4, anchor=west, fill=#1] (#3) at (#5.east) {\entrysize #2};
  }

  \newcommand{\drawBitSpan}[2]{
    \draw[|-|] ([yshift=\arrowYOffset]#1.west) -- ([yshift=\arrowYOffset]#1.east)
    node[midway, above, align=center] {\bitsize #2};
  }

  \newcommand{\secondLevelArrow}[3]{
    \draw[|-|] ([yshift=2*\arrowYOffset]#1.west) -- ([yshift=2*\arrowYOffset]#2.east)
    node[midway, above, align=center] {\bitsize #3};
  }
  \newcommand{\thirdLevelArrow}[3]{
    \draw[|-|] ([yshift=3*\arrowYOffset]#1.west) -- ([yshift=3*\arrowYOffset]#2.east)
    node[midway, above, align=center] {\bitsize #3};
  }

  \newcommand{\mybrace}[4]{
	\draw [ thick, decoration={ brace, raise=#4}, decorate] (#1.east) -- (#2.west) 
	node [pos=0.5,anchor=north,yshift={-#4 - 0.05cm}, font=\bitsize,align=center] {#3}; 
  }

  \def\Lonecolor{my-blue} 
  \def\Ldeltacolor{my-green}
  \def\LCcolor{my-orange}
  \def\Scolor{my-yellow}
  \def\Dcolor{yellow!30}
  \def\Xcolor{gray!20}
  \def\minbitcolor{my-red}
  \def\Lzerowidth{23}

  \node[cacheentry, minimum width=35, fill={\Lonecolor}] (L1) at (0,0) {\entrysize L1};
  \mynewbox{\Scolor}{$S$}{S}{20}{L1}
  \mynewbox{\Dcolor}{$D$}{D}{20}{S}
  \mynewbox{\Ldeltacolor}{L0$^{\Delta}$}{L0a}{\Lzerowidth}{D}
  \mynewbox{\Ldeltacolor}{$...$}{L0b}{\Lzerowidth}{L0a}
  \mynewbox{\Ldeltacolor}{L0$^{\Delta}$}{L0c}{\Lzerowidth}{L0b}
  \mynewbox{\LCcolor}{L0$^C$}{L0d}{\Lzerowidth}{L0c}
  \mynewbox{\LCcolor}{$...$}{L0e}{\Lzerowidth}{L0d}
  \mynewbox{\LCcolor}{L0$^C$}{L0f}{\Lzerowidth}{L0e}

  \drawBitSpan{L1}{43}
  \drawBitSpan{S}{37}
  \drawBitSpan{D}{32}

  \drawBitSpan{L0a}{12}
  \drawBitSpan{L0c}{12}
  \drawBitSpan{L0d}{16}
  \drawBitSpan{L0f}{16}

  \secondLevelArrow{L1}{D}{112 bits}
  \secondLevelArrow{L0a}{L0c}{$24 \cdot 12$ bits }
  \secondLevelArrow{L0d}{L0f}{$7 \cdot 16$ bits }

  \def\bvcolor{gray!40}
  \def\bvsize{\tiny}
  \def\bvwidth{8}
  \newcommand{\bvbox}[4]{
      \node[bventry, minimum width=#3, anchor=west, fill=\bvcolor, font=\bvsize] (#2) at (#4.east) {#1};
  }
  \node[bventry, minimum width=30, fill=\bvcolor, anchor=north west, font=\bvsize] (BV1) at (-2,-1.5) { $101001...$};
  \bvbox{$010111...$}{BV2}{\bvwidth}{BV1}
  \bvbox{$110100...$}{BV3}{\bvwidth}{BV2}
  \bvbox{$111000...$}{BV4}{10}{BV3}
  \node[anchor=north west, font=\Large] (dots) at (BV4.north east) {$\cdots$};
  \bvbox{$001000...$}{BV5}{10}{dots}
  \bvbox{$101111...$}{BV6}{10}{BV5}

  \fill[gray!20] (L1.south west) -- (L1.south east) -- (BV6.north east) -- (BV1.north west) -- cycle;

  \draw (L0a.south) -- (BV1.north);
  \draw (L0b.south west) -- (BV2.north);
  \draw (L0b.south) -- (BV3.north);
  \draw (L0d.south) -- (BV4.north east);

  \def\belowfactor{0.8}
  \newcommand{\drawBelowSpan}[2]{
    \draw[|-|] ([yshift=-\belowfactor*\arrowYOffset]#1.west) -- ([yshift=-\belowfactor*\arrowYOffset]#1.east)
    node[midway, below, align=center] {\bitsize #2};
  }
  \newcommand{\drawBelowSpanSecond}[3]{
    \draw[|-|] ([yshift=-2*\belowfactor * \arrowYOffset]#1.west) -- ([yshift=-2*\belowfactor * \arrowYOffset]#2.east)
    node[midway, below, align=center] {\bitsize #3};
  }
  \node[below=1pt, font=\scriptsize, align=left] at (L1.south) {covers \\ 65536 bits};
  \drawBelowSpan{BV1}{2048}
  \drawBelowSpan{BV2}{2048}
  \thirdLevelArrow{L1}{L0f}{512 bits}
  \drawBelowSpanSecond{BV1}{BV6}{$2^{16}=32 \cdot 2048=65536$ bits}

\end{tikzpicture}

\caption{Structure of a summary tree L1-Block for \(L_0=2048\) and \(L_1=2^{16}\).
  Different parts of the L1-block are colored.
The underlying bit vector is shown in gray.\nr{If we want the Bitvector-Compression shown in the plots, we can also include here the $S$ and $D$ pointer, and say 2-3 sentences on the whole compression business}}

\end{figure}
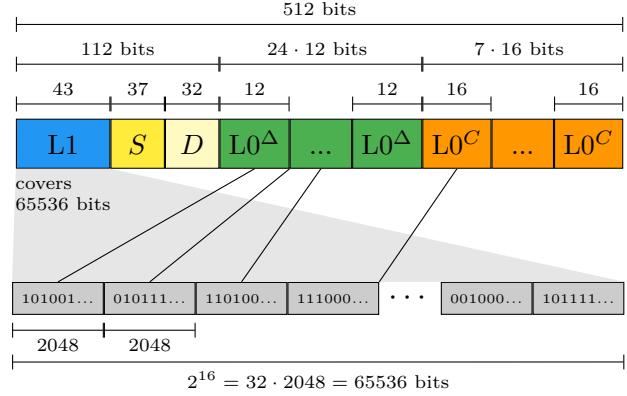
In each L1-block, we aggregate 32 L0-blocks with aggregation factor $d_1=32$ and L1-block size $L_1=32\cdot L_0 \in \{2^{14},2^{15},2^{16}\}$.
The three \nr{four} parts of our L1-block are shown in \cref{L1-block}.

The value L1 (not to be confused with the size of an L1-block \(L_1\)) stores the number of \texttt{1}-bits occurring before the L1-block.
L1 $\in 0..n$ is monotonically increasing.

Seven values L0$^C$ store the cumulative count of \texttt{1}-bits from the start of the L1-block to the end of their respective L0-block.
They are stored for every 4-th L0-block.
For the 32-th L0-block no such value has to be stored, as it is already encoded in the next L1-value.

24 values L0$^\Delta$ store the individual count of \texttt{1}-bits within each of the remaining 24 L0-blocks.
This concludes the contents of an L1-block.

Preliminary experiments have shown that compression of the underlying bit vector can improve query performance.
To this end, we have implemented a basic compression technique where sufficiently sparse L0-blocks ($\leq 10\%$ density) are stored via delta encoding.
To find these sparse and dense encodings per L1-block we store two pointers $S$ and $D$.
Due to space constraints and the preliminary nature of the  implementation of the compression, we leave a more comprehensive study for future work.
However, we hint at its effectiveness in our experimental evaluation in \cref{sec:experimental_evaluation}.

We use the same L1-block layout for all three choices of $L_0$.
 At the largest configuration, $L_0=2048$, the L1-block covers $L_1=65536$ bits.
The largest L0$^C$ value stores the number of \texttt{1}-bits up to the end of the 28-th L0-block, which can then be at most 57344, which can be encoded in 16 bits.
Similarly, each L0-block contains between 0 and 2048 \texttt{1}-bits, i.e., 2049 distinct values.
This can be represented in 12 bits by each L0$^\Delta$ value. All L0 values together consume thus 400 bits. 112 bits remain to fit into exactly one cache line (512 bits), they are distributed between L1, $D$ and $S$.

We save 96 bits by interleaving individual L0$^\Delta$ counts with cumulative L0$^C$ counts.
This comes at the cost of reconstructing the full cumulative information at run time.
Fortunately, this can efficiently be done with SIMD instructions.
Briefly, the L0$^C$ and L0$^\Delta$ values can be read as four \texttt{\_\_m128i} vectors and then added together three times.
This way,  all 32 cumulative counts can be reconstructed in only 9--10 CPU cycles.
This makes use of the AVX2 instruction set, which is available on nearly all commercially available x86-based CPUs since 2013.
Not using such vector instructions would neglect key capabilities of the current standard CPU architectures. 

The complete 2-level summary tree is obtained by storing the L1-blocks consecutively in a single array.
We can represent two summary tree levels in a single array, because we have carved out enough space around the L0 values to store the L1-value directly next to them.

The L2-blocks in our 3-level summary-tree aggregates 16 L1-blocks, i.e., $L_2 = 16\cdot L_1$.
The L1-entries in the L2-block store the number of \texttt{1}-bits from its start up to the end of their respective L1-block.
Each L1-value is stored in 32 bits for fast SIMD readout, up to the first, which is stored in 17 bits, next to 47 bits for the L2-value. In total, each L2-block requires 512 bits.

The overhead of our summary tree implementation is $512/L_1$ for $k=2$.
This results in 3.1\,\%, 1.56\,\%, and 0.78\,\% overhead for $L_0 = 512$, $L_0=1024$, and $L_0=2048$, resp.
The value $L_0$ thus offers a space/time trade-off. For $k=3$ the L2-blocks cause additional overhead of $512/(16\cdot L_1)$, increasing the overheads by factor of 17/16. 

We now describe a select$_1(i)$ query in the 2-level summary tree ($k=3$ works analogously, first looking at L2-blocks).
Starting with the $(\flOffset + \slOffset + \tlOffset)$-th L1-block, we scan trough the L1-values until $i<\textnormal{L1}$ marks the correct L1-block.
We scan at most $\alpha\in O(1)$ L1-values.
In the correct L1-block, we use the aforementioned SIMD instructions to convert the L0$^C$- and L0$^\Delta$-values into 32 cumulative L0-values.
We compare them to $i -\textnormal{L1}$ via \texttt{\_mm256\_cmpgt\_epi16}.
This results in a vector with \texttt{1}-bits set for each L0-block occurring before the correct L0-block.
We extract the \texttt{1}-bits via \texttt{\_mm256\_movemask\_epi8} and count them with \texttt{popcnt}.
This gives us the the correct L0-block index.
In this L0-block, we scan $\leq L_0$ bits (using popcount) until we find the 64-bit word containing the sought bit.
We obtain its index in the final word by using \texttt{\_pdep\_u64} and \texttt{\_tzcnt}~\cite{PandeyBJ2017FastSelect}.

\begin{table}[t]
	\caption{Practical implementations of rank-select data structures.
    In addition to the asymptotic query times, we give the number of sample tree (\(\ell\)) and summary tree (\(k\)) levels.
  Not all data structures support all queries.}
\label{tab:overview_data_structures}
\centering
\begin{tabular}{lllll}
  \toprule
  & \multicolumn{2}{c}{\textbf{query time \((O(\cdot))\)}}
  & \multicolumn{2}{c}{\textbf{levels}} \\
  \cmidrule{2-5}
\textbf{name} & select$_1$ & rank & \(\ell\)  & \(k\) \\
\midrule
sdsl-1/5~\cite{GogBMP2014SDSL} & \footnotesize{---} &  $1$  & \footnotesize{---} & \footnotesize{---} \\
\midrule
sdsl-mcl~\cite{GogBMP2014SDSL} & $1$ & \footnotesize{---}  & 2 & 2 \\
simple-\(\{0..3\}\)~\cite{Vigna2008BroadwordRankSelect} & $\log\log n$ & \footnotesize{---} &  2 & 1\\
sdsl-sd~\cite{GogBMP2014SDSL} & $\log n$ & \footnotesize{---} & 1 & 2 \\
\midrule
rank9select~\cite{Vigna2008BroadwordRankSelect} & $1$  & $1$  & 2 & 2\\
poppy~\cite{ZhouAK2013PopcountRankSelect}& $n$  &  $1$     & 1 & 2 \\
Pasta Flat~\cite{Kurpicz2022PastaFlat} & $n$ & $1$ & 1 & 2 \\
\textbf{\(\bm \nameDS\)} [here]& $1$ & $1$  & 3 & 2 \& 3 \\
\bottomrule
\end{tabular}
\end{table}

\section{Experimental Evaluation}
\label{sec:experimental_evaluation}
We ran all of our experiments on a machine equipped with an Intel i7 11700 processor (8 cores/16 hardware threads), a base clock speed of 2.5 GHz, and support for AVX-512) and 64 GiB RAM.
The machine ran with Rocky Linux 9.4. The program was written in C++ and compiled with flags \texttt{-DNDEBUG -march=native -O3} with Clang 19.1.0.

We used two kinds of inputs.
First, we tested all data structures on uniformly distributed bits with different densities, i.e., fill-rates between 50\,\% and 1\,\%.
Even sparser bit vectors allow for significant compression, which is outside the scope of this paper. 
The bit vectors have lengths \(1\cdot 10^8\), \(8\cdot 10^8\), and \(64\cdot 10^8\).
We also evaluate a difficult setting for select queries, where we ask for the position of the \texttt{1}-bit right after \(10^d\) \texttt{0}-bits for \(d\in3..8\).
All reported times are the average over 20\,M queries.
All experiments were conducted sequentially.

\subsection*{Implementations}
We compared our implementation (\cref{sec:implementation}) with the following existing rank and/or select data structures, see \cref{tab:overview_data_structures} for more details.\footnote{We also have an implementation with an uncompressed sample tree. We do not show it herer since it is outperformed by \nameDS.}
\begin{description}[noitemsep,topsep=0pt,parsep=0pt,partopsep=0pt]
\item[sdsl-v1/2~\cite{GogBMP2014SDSL}:] Constant time rank only data structures requiring 25\,\% and 6.25\,\% additional space.
\item[sdsl-mcl~\cite{GogBMP2014SDSL}:] Implementation of original \(o(n)\) space select only data structure (25\,\% space-overhead)~\cite{ClarkM1996Select}.
\item[sdsl-sd~\cite{GogBMP2014SDSL}:] Fast compressed bit vector with select support. When no compression is possible up to 80\,\% space-overhead.
\item[simple-0/1/2/3~\cite{Vigna2008BroadwordRankSelect}:] Fastest select only data structure using broadword programming. Offers a space-time trade-off with a space-overhead between 1.75\,\% and 11\,\%. Has no rank support.
\item[rank9select~\cite{Vigna2008BroadwordRankSelect}] Constant time rank and select data structure with 50\,\% space-overhead.
\item[poppy~\cite{ZhouAK2013PopcountRankSelect}:] A rank and select (non-constant time) data structure that requires \(3.5\,\%\) additional space. Unfortunately, we were only able to run it successfully on bit vectors smaller than the ones we tested here.
\item[Pasta Flat~\cite{Kurpicz2022PastaFlat}:] The currently fastest and most space-efficient rank and select (non-constant time) data structure (\(3.5\,\%\) space-overhead). Uses SIMD.
\end{description}

\begin{figure*}[t]
	\centering
	\includegraphics[scale=.8]{\detokenize{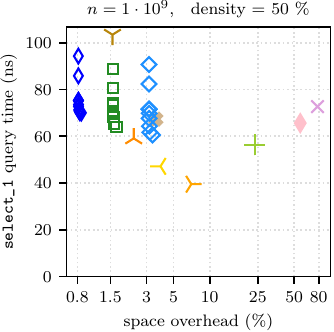}}	\includegraphics[scale=.8]{\detokenize{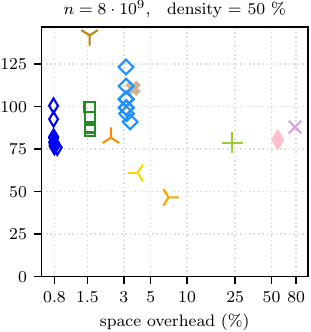}}	\includegraphics[scale=.8]{\detokenize{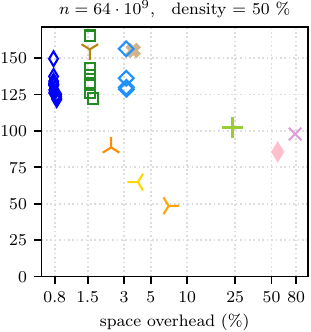}} \\

	\hspace{30pt}\includegraphics[]{\detokenize{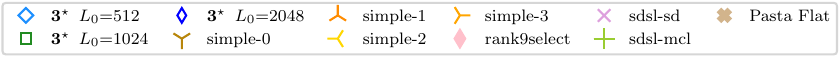}}

	\label{fig:pareto-plot}
	\caption{Pareto plots for the independent-bits instances and density 50\,\%. For other densities, see \cref{fig:other_pareto_plots}.}
\end{figure*}

\subsection{Alternative Sample Tree Parameters}
\label{a_b_discussion}
The values \astar\ and \bstar\ minimize a bound on the sample tree size.
However, optimizing for size alone might not always be the best practical choice, 
in particular as accesses into the mid- or even bot-level can cause cache-misses.
We can optimize for speed by considering smaller values for $a$ that sample more dense than $\alpha L$, e.g., $\afast= \frac{\alpha Lm}{3n}$.
Random queries will then rarely have to look into the mid- and bot-level, allowing faster queries---at the expense of space.
In the worst case, the top-level will dominate in space, with size $ \leq \frac{3n}{\alpha L} \log\frac{n}{L} \in O(n \frac{\log \log n}{\log n})$.
For a given instance, we can also test multiple values $(a,b)$ to find the specific $(a_{\textnormal{min}},b_{\textnormal{min}})$ that provide the smallest tree.

\subsection{Determine Parameters for \nameDS}
Our implementation of \nameDS\ has parameters \(L,a,b\) and \(\alpha\) for which we can find optimal, e.g., \(a^\star\) and \(b^\star\) values, but also experimentally determine values.
In \cref{fig:pareto-plot}, we give a space-time Pareto plot \ps{added explanation of Pareto}(i.e. showing configurations that are not dominated within their family) for different configurations of \nameDS\ for density 50\,\%.
For other densities, which give very similar results, see \cref{fig:other_pareto_plots}.

We can see an interesting dependency on \(L_0\).
On small inputs, we achieve the best query performance for \(L_0=512\) as this saves
work for population count instructions on the lowest level.
Surprisingly, for larger inputs \(L_0=2048\) is better.
Likely, the smaller sample tree reduces cache faults while traversing the sample tree, offsetting overhead for fetching larger L0-blocks which may be fast due to hardware prefetching.

Considering our competitors, only simple-0/1/2/3 has a tuning parameter.
It gives us a clearly visible space-time trade-off.
On the smallest input, simple-1 dominates \nameDS (albeit without support for rank).
On the longer inputs, \nameDS\ and simple-1/2/3 are both on the Pareto front \ps{added explanation},i.e., dominated by other data structures.
All other competitors are also not on the Pareto front, mostly because they require significantly more memory without a significant improvement in performance compared to the more space-efficient implementations.
Note that on specific instances, these implementation perform well, see \cref{sec:experiments_query_time}.

\paragraph{Selected Parameters}
From the Pareto plots, we chose four configurations for \nameDS: First, ($L_0=512, a_\textnormal{fast}, \alpha=2, k=2$) performs especially well on small instances.
Second, ($L_0=2048, a^*, \alpha=16, k=2$) is robust on all inputs. 
Third, ($L_0=2048, a_\textnormal{fast}, \alpha=8, k=3$) performs especially well on large instance.
And fourth, ($L_0=2048, a_\textnormal{min}, \alpha=32, k=2$) has a very small overhead, e.g., for independent bits at 50\,\% density its sample tree requires 94\,kB for a 8\,GB instance.
With a $k=3$ summary tree, its sample tree shrinks down to 6\,kB, but the query time also deteriorates significantly.

\begin{figure*}[t]
	\centering
	\includegraphics[scale=.8]{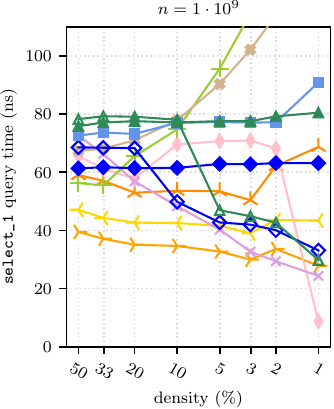}
	\includegraphics[scale=.8]{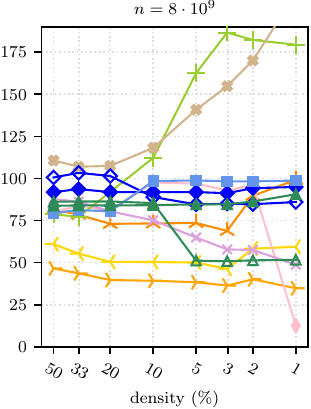}
	\includegraphics[scale=.8]{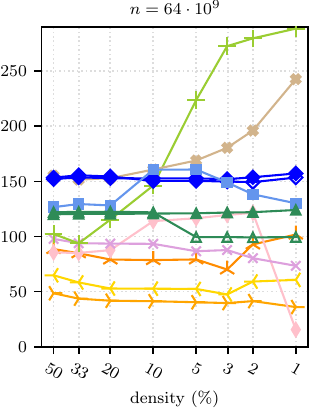}

	\vspace{5pt}
	\hspace{30pt}\includegraphics[]{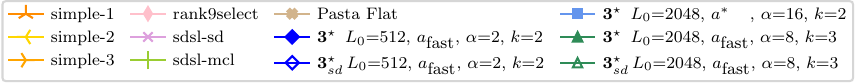}
	\caption{Select query times for different densities.
          Three Pareto-optimal configurations are shown for \nameDS.
        }
        \label{fig:indep-bits-curve}
      \end{figure*}

\subsection{Select Query Times}
\label{sec:experiments_query_time}
In \cref{fig:indep-bits-curve}, we show the query time of our selected configurations and our competitors for different input sizes and densities.

It is apparent that simple-2/3 outperform \nameDS\ on all inputs.
That is, unless we enable compression, see~\cref{sec:implementation} for more details on our prototypical compression scheme.
Our compressed variants are marked by the term \emph{sd} and help us outperform simple-2 on densities at most 5\,\%.
Note that even in the uncompressed case simple-2/3 require \(4.5\times\) and \(9.6\times\) more space and do not support rank queries.
All non-compressed configurations behave as expected, having a query time independent of the density.

Pasta Flat and (surprisingly) sdsl-mcl require more time as the density decreases.
For Pasta flat, this can be explained by the fixed sample rate (every 8192 \texttt{1}-bit is sampled) that require the data structure to scan through a large range of superblocks.
We are not sure why also sdsl-mcl behaves this way.
Also an outlier is rank9select, which becomes very fast for density 1\,\%, as it can sample most of the positions here directly.
However, this comes at the cost of a 50\,\% space-overhead.

\begin{figure}[t]
	\centering
	\includegraphics[]{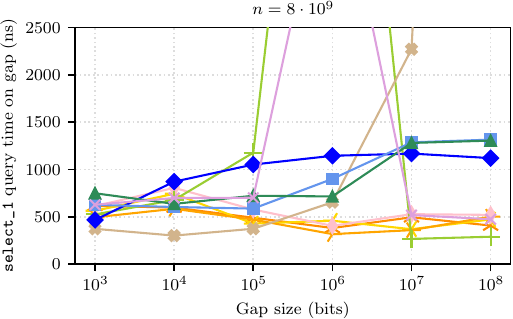}

	\vspace{5pt}
	\hspace{30pt}\includegraphics[]{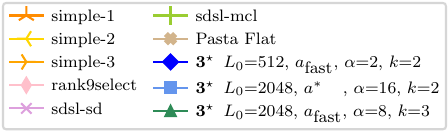}
	\label{fig:gap}
	\caption{Difficult select instances with a gap consisting of only \texttt{0}-bits in the middle. Query times for the first \texttt{1}-bit after the gap.}
\end{figure}

\subsection{Difficult Select Queries}
In \cref{fig:gap}, we report the query times for our difficult gap-instance.
Here, we see that the space-inefficient and select-only implementations can cope with large gap sizes.
Space-efficient rank \emph{and} select data structures other than \nameDS\ have outliers for some gap sizes, as their sample strategies cannot cope with this kind of input.

\subsection{Rank Queries}
As our main contribution are for select data structures and due to space-constraints, the results for rank queries can be found in \cref{sec:additional_experimental_results}.

\section{Conclusion and Outlook}
\label{sec:conclusion_and_outlook}
We have brought theory and practice of rank and select data structures much closer together. This can impact a wide variety of applications. We want to explore more trade-offs between space and performance. In particular, we believe that faster queries with only slightly more space overhead is possible. The possibility to  compress the bit vector itself opens up a further interesting dimension in the design space. In particular, it is intriguing that large gaps that can be responsible for expensive select queries can be highly compressed.

\bibliography{references}

\clearpage

\appendix

\section{Symbols and Notation}
~\\[-.5cm]
\begin{table}[h!]
  \caption{Notation used in this paper.}
  \label{tab:notation}
  \centering
  \begin{tabular}{ll}
    \toprule
    symbol & meaning \\
    \midrule
    $v$ & a bit vector\\
    $n$ & bit-vector size\\
    $m$ & number of \texttt{1}-bits\\
    $\ell$ & number of sample-tree levels\\
    $k$ & number of summary-tree levels\\
    $L_i$ & level $i\in 0..k-1$ summary-tree block size\\
    $L$ & superblock size $L_{k-1}$\\
    $a_i$ & sample rate at level $i\in 0..\ell-1$\\
    $a$ & $a_0$\\
    $b$ & $a_1$\\
    $\alpha$ & scan threshold\\
    $\flOffset,\slOffset,\tlOffset$ & superblock offset (top, mid, and bot)\\
    $g$ & index of mid-group\\
    $\kappa$ & size of max. counter in mid-group\\
    $\rho$ & size of $\tlOffset$ for entry\\
    $c$ & \# gaps w/ same size \(\rho\) in mid-group \\
    \bottomrule
  \end{tabular}
\end{table}

\section{Bit-Compressed 2-Level Sample Tree}\label{s:compressedTwo}
We outline and analyze an approach \nameDStwo\ that is a simplification of \nameDS.
We use an uncompressed top-level sample of every $a$-th \texttt{1}-bit.
More precisely, an entry comprises a $\log\frac{n}{L}$-bit offset $o$
and a $\log s$-bit group number $g$.
For up to $s=\frac{n}{\Nmax L}$ large gaps, we store dense samples
that need $a\log r$ bits for gaps of size $r>\Nmax$.

Thus, the overall space is
\begin{align*}
  S &=\frac{m}{a}\left(\log\frac{n}{L}+\log\frac{n}{\Nmax L}\right)+\sum_{i=1}^sa\log r_i\\
  &\leq 2\frac{m}{a}\log\frac{n}{L}+sa\log\alpha
  = 2\frac{m}{a}\log\frac{n}{L}+\frac{an}{\Nmax L}\log\alpha.
\end{align*}
This expression is minimized for $a=\sqrt{\frac{2m\log(n/L)}{n\log\alpha}}$.
Substituting this back into $S$ yields
$$S=\sqrt{\frac{8mn\log(n/L)\log\alpha}{\alpha L}}=O\left(n\sqrt{\frac{\log n}{L}}\right)$$
For $L=\Theta(\log^k(n)/\log\log^{k-1}n)$ we get space
$$S=O\left(n\frac{\log\log^{(k-1)/2}n}{\log^{(k-1)/2}n}\right).$$

\section{Space of Bit-Compressed 3-Level Sample Tree with a 2-Level Summary Tree}
\label{sec:special_case_space}
We now show that \(\nameDS_\textnormal{space}\in O(n\frac{\sqrt{\log\log n}\sqrt{\log\log\log n}}{\log n})\) for \(\ell=3\) and \(k=2\).
Most importantly, here, we have \(G=t=\frac{m}{a}\) if we want to minimize the space.

This gives us new values for \(a^\star\) and \(b^\star\), which we can derive with similar arguments as in the main part of the paper.

\begin{lemma}
  \label{lem:optimal_a_and_b}
  To minimize the space, the optimal choices for \(a\) and \(b\) are:

  \[a^\star= \sqrt{\frac{m}{2n} \frac{\alpha}{\log \alpha} L} \cdot \log k \cdot \sqrt{\log \frac{a^*n}{Lm}}\] and \[b^\star = \sqrt{2 \frac{m}{n} \frac{\alpha}{\log \alpha}L \log \left(\frac{an}{Lm}\right)}\]

  \noindent Asymptotically, we have
  \[a^\star\in O(\log\frac{n}{L}(L\log\log\log n)^{1/2})\]
  and
  \[b^\star\in O((L\log\frac{a}{L})^{1/2}).\]
  
\end{lemma}
\nr{$a^*$ on purpose in recursive form, on purpose? I would have put the approximated but explicit $a^*$ here, and show the true recursive dependency only in the derivation}

Similar to the result in the main part, we can now give an upper bound for the space required by \nameDS\ for \(\ell=3\) and \(k=2\).

\begin{theorem}
 Our 3-level sample tree \nameDS\ using a \(2\)-level summary tree requires \[O\left(n\frac{\sqrt{\log\log n}\sqrt{\log\log\log n}}{\log n}\right)\] bits of space.
\end{theorem}

\begin{proof}
  We now combine \cref{lem:with_a_and_b} with \cref{lem:optimal_a_and_b} to show the worst-case space-requirements of our data structure.
  To this end, we show that all addends meet the upper bound.

  In addition, we use the following upper bounds:
  \(t=m/a=O(n/a)\), \(L=O(\log^2 n/\log\log n)\), \(\frac{a^\star}{b^\star}=O(\log\frac{n}{L})\), \(b^\star=O(\frac{\sqrt{\log\log n}\sqrt{\log\log\log n}}{\log n})\), \(\frac{k}{t}=O(\frac{a}{L})=O(\sqrt{\log\log n}\sqrt{\log\log\log n})\), \(k=O(\frac{n}{\alpha L})=O(\frac{n}{L})\)

\begin{description}
\item[\(t(\log\frac{n}{L}+\log t+\log\log\frac{a}{b})\):] Here, we individually consider the three addends:
  \begin{enumerate}
  \item \(t\log\frac{n}{L}=O(\frac{n}{a}\log\frac{n}{L})=O(\frac{n}{\log\frac{n}{L}(L\log\log\log n)^{1/2}}\log\frac{n}{L})=O(n\frac{\sqrt{\log\log n}}{\log n\sqrt{\log\log\log n}})\)
  \item \(t\log t=O(\frac{n}{a}\log\frac{n}{a})=O(\frac{n}{b^\star}\frac{\log\frac{n}{b^\star\log\frac{n}{L}}}{\log\frac{n}{L}})=O(n\frac{\sqrt{\log\log n}}{\log n}\frac{\log\frac{n}{b^\star\log\frac{n}{L}}}{\log\frac{n}{L}})=O(n\frac{\sqrt{\log\log n}}{\log n})\)
  \item \(t\log\log\log\frac{a^\star}{b^\star}=O(\frac{n}{b^\star\log\frac{n}{L}}\log\log\log\frac{n}{L})=O(\frac{n}{b^\star})=O(n\frac{\sqrt{\log\log n}}{\log n})\)
  \end{enumerate}
\item[\(t \left( 2\frac{a}{b}\log \frac{\alpha k}{t}+ \log \log \frac{\alpha k}{t} + \log t \right)\):] Again, we consider the three addends individually:
  \begin{enumerate}
  \item \(t2\frac{a^\star}{b^\star}\log\frac{\alpha k}{t}=O(\frac{n}{a^\star}\frac{a^\star}{b^\star}\log\frac{k}{t})=O(\frac{n}{b^\star}\log\frac{a}{L})=O(n\frac{\sqrt{\log\log n}}{\log n\sqrt{\log\log\log n}}\log\log\log n)=O(n\frac{\sqrt{\log\log n}\sqrt{\log\log\log n}}{\log n})\)
  \item \(t\log\log\frac{\alpha k}{t}=O(\frac{n}{a^\star}\log\log\log\log n)=O(n\frac{\sqrt{\log\log\log n}}{b^\star\log\frac{n}{L}})=O(n\frac{\sqrt{\log\log n}\sqrt{\log\log\log n}}{\log n\log\frac{n}{L}})\)
  \item We already did \(t\log t\).
  \end{enumerate}
\item[\(t \left( \log \frac{k}{t} + 1 \right) \log k\):] Again, we consider each addend individually:
  \begin{enumerate}
  \item \(t\log\frac{k}{t}\log k=O(\frac{n}{b^\star\log\frac{n}{L}}\log\frac{a}{L}\log\frac{n}{L})=O(n\frac{\log\frac{a}{L}}{b^\star})=O(n\frac{\sqrt{\log\log n}\sqrt{\log\log\log n}}{\log n})\)
  \item \(t\log k=O(\frac{n}{b^\star\log\frac{n}{L}}\log\frac{n}{L})=O(\frac{n}{b^\star})=O(n\frac{\sqrt{\log\log n}}{\log n})\)
  \end{enumerate}
  \item[\(kb \log \alpha\):] Here, we have \(kb^\star\log\alpha=O(kb^\star)=O(\frac{n}{L}\sqrt{L\log\log\log n})=O(n\frac{\sqrt{\log\log\log n}}{\sqrt{L}})=O(n\frac{\sqrt{\log\log n}\sqrt{\log\log\log n}}{\log n})\)
\end{description}
  
\end{proof}

\section{Additional Experimental Results}
\label{sec:additional_experimental_results}
In this section, we present additional experimental results.
First, in \cref{fig:other_pareto_plots}, we give the Pareto plots for smaller densities (10\,\% and 1\,\%).
Here, the overall picture is the same.
The most notably difference is the very impressive query performance of rank9select.
For these inputs, it is the overall fastest select data structure, only being penalized by its 50\,\% space-overhead.
The relative behavior of \nameDS\ and simple-0/1/2/3 remains the as in the main part of the paper.

Next, in \cref{fig:rank}, we give rank query times for all data structures supporting rank queries.
Here, the most notable difference, compared to all other plots, is the occurrence of poppy for small inputs.
This are the only benchmarks it successfully concluded.
When it comes to rank queries, \nameDS\ is competitive but never the fastest rank data structure.
On larger bit vectors it even is the slowest one.
However, it should be noted that \nameDS\ is the the most space-efficient rank data structure.
In \cref{fig:rank-pareto}, we give a space/time plot that shows how much less memory it requires.
\nameDS\ requires \(4.5\times\) less space than the current state-of-the-art space-efficient rank and select data structure Pasta Flat.
On the extreme side, it requires \(64.1\times\) less space than the fastest rank and select data structure---while being only \(3.3\times\) to \(4.8\times\) slower.

\begin{figure*}[h]
	\centering
	\includegraphics[scale=.8]{\detokenize{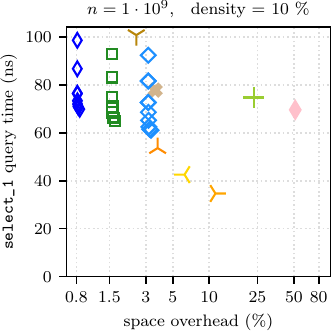}}
\includegraphics[scale=.8]{\detokenize{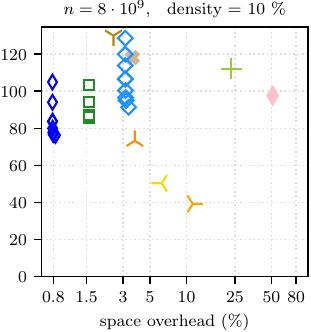}}
\includegraphics[scale=.8]{\detokenize{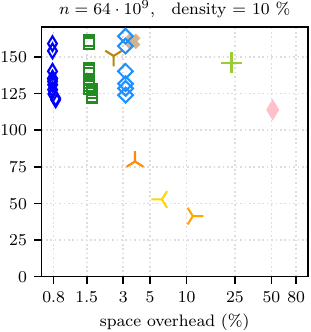}}
	\vspace{20pt}

	\includegraphics[scale=.8]{\detokenize{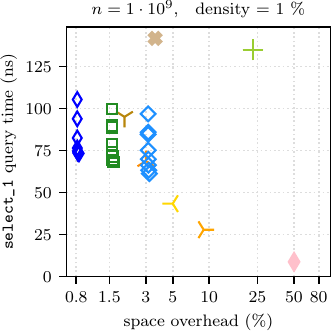}}
\includegraphics[scale=.8]{\detokenize{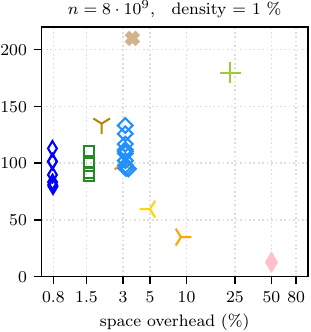}}
\includegraphics[scale=.8]{\detokenize{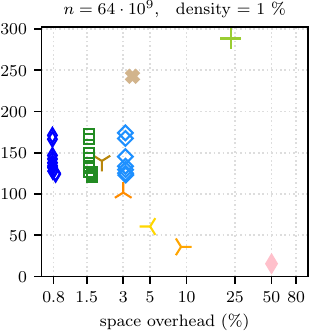}}
	\vspace{10pt}

	\includegraphics[]{\detokenize{plots/pareto2/_legend_.pdf}}
	\caption{Remaining Pareto plots for smaller densities.}
        \label{fig:other_pareto_plots}
\end{figure*}

\begin{figure*}[t]
	\centering
	\includegraphics[scale=.8]{\detokenize{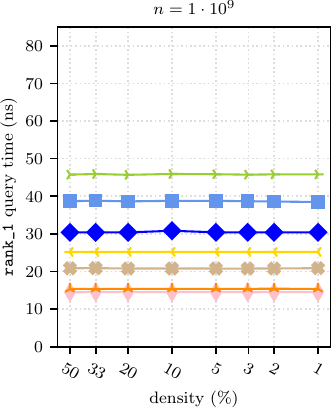}}
\includegraphics[scale=.8]{\detokenize{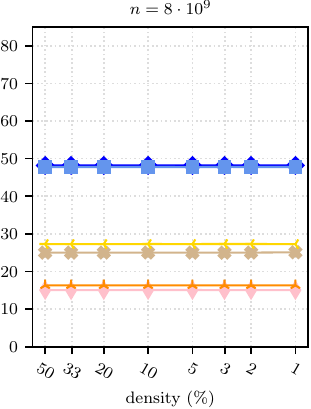}}	\includegraphics[scale=.8]{\detokenize{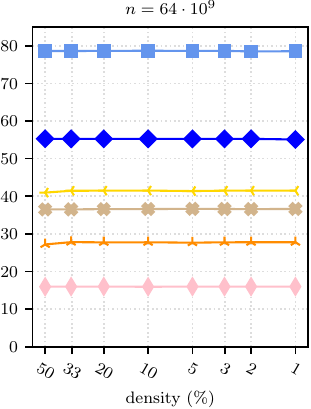}}
	\hspace{30pt}\includegraphics[]{\detokenize{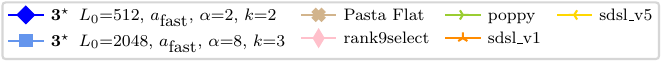}}
	\caption{Query times for rank queries on bit vectors with different densities and uniform distribution of \texttt{1}-bits.}
        \label{fig:rank}
\end{figure*}

\begin{figure*}[t]
	\centering
	\includegraphics[scale=.8]{\detokenize{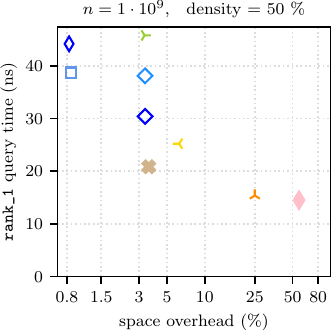}}
\includegraphics[scale=.8]{\detokenize{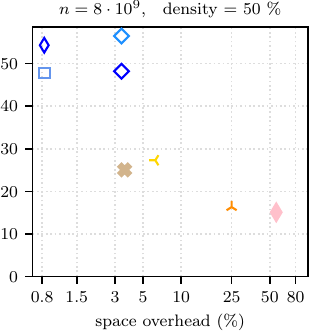}}	\includegraphics[scale=.8]{\detokenize{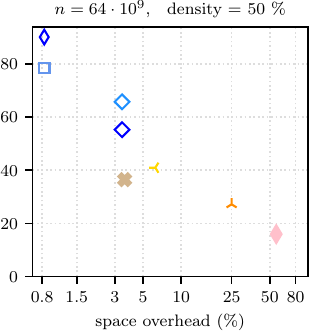}}
	\hspace{30pt}\includegraphics[]{\detokenize{plots/rank/_legend_.pdf}}
	\caption{Pareto plots for space overhead and. rank query time on independent-bit instances with density 50\%.}
        \label{fig:rank-pareto}
\end{figure*}

\end{document}